%% file: CFcopula_202303011.tex
\documentclass[12pt]{article}

\input{preamble}
\usepackage{comment}
\newcommand{\noop}[1]{}

\title{Counterfactual Copula and Its Application to the Effects of College Education on Intergenerational Mobility\thanks{\protect\linespread{1}\protect\selectfont We would like to thank Yi-Ting Chen, Chih-Sheng Hsieh, Yu-Chin Hsu and participants at the 2022 Annual Meeting of Taiwan Econometric Society for their valuable comments. Jiun-Hua Su gratefully acknowledges research support from the Ministry of Science and Technology, Taiwan (MOST 107-2410-H-001-112-MY2).}}

\author{
Tsung-Chih Lai\thanks{Department of Economics, National Chung Cheng University. Email: tclai@ccu.edu.tw}
\and
Jiun-Hua Su\thanks{Corresponding author. Institute of Economics, Academia Sinica. Email: jhsu@econ.sinica.edu.tw}
}

\date{}

\begin{document}

\maketitle

\begin{abstract}
This paper proposes a nonparametric estimator of the counterfactual copula of two outcome variables that would be affected by a policy intervention.
The proposed estimator allows policymakers to conduct \emph{ex-ante} evaluations by comparing the estimated counterfactual and actual copulas as well as their corresponding measures of association.
Asymptotic properties of the counterfactual copula estimator are established under regularity conditions.
These conditions are also used to validate the nonparametric bootstrap for inference on counterfactual quantities.
Simulation results indicate that our estimation and inference procedures perform well in moderately sized samples.
Applying the proposed method to studying the effects of college education on intergenerational income mobility under two counterfactual scenarios,
we find that while providing some college education to \emph{all} children is unlikely to promote mobility, offering a college \emph{degree} to children from less educated families can significantly reduce income persistence across generations.

\bigskip



\end{abstract}

\newpage

\section{Introduction}
\label{Introduction}
Counterfactual analysis is an important method in policy evaluation, especially when neither a policy nor its quasi-experiment has yet been implemented.
Recent studies, for example \citet{Rothe2010}, \citet{ChernozhukovFernandez-ValEtAl2013a}, and \citet{Hsu2022}, focus on the counterfactual distribution of a scalar outcome variable affected by an exogenous policy intervention.
However, a change in a policy variable may also cause a change in dependence structure between two outcome variables of interest.
For example, in the intergenerational mobility literature, college education has been regarded as a ``great equalizer'' to reduce the correlation between incomes across generations by offering equal opportunities to children regardless of the advantages of origins (\citeauthor{Torche2011}, \citeyear{Torche2011} and \citeauthor{Chetty2020}, \citeyear{Chetty2020}).
Accordingly, policymakers would wish to evaluate \emph{ex ante} to what extent income mobility could be improved if the higher education system were further expanded.

Motivated by such counterfactual policy evaluation, we are concerned with its effect on a bivariate copula function of the two outcome variables
\begin{align*}
Y_1=g_1(X,\varepsilon_1)\;\;\text{and}\;\;
Y_2=g_2(X,\varepsilon_2),
\end{align*}
where $(Y_1,Y_2)^\top$ is a two-dimensional vector of outcome variables that are continuously distributed, $X$ is a $d$-dimensional vector of covariates, $(\varepsilon_1,\varepsilon_2)^\top$ is individual unobserved heterogeneity in an arbitrary measurable space of unrestricted dimensionality, and $(g_1,g_2)^\top$ are functions that are unknown to researchers.
Policymakers attempt to exogenously manipulate the covariate from $X$ to $X^*$, and this manipulation leads to the counterfactual outcome variables
\begin{align*}
Y_1^{*}=g_1(X^{*},\varepsilon_1)\;\;\text{and}\;\;
Y_2^{*}=g_2(X^{*},\varepsilon_2),
\end{align*}
but the functions $(g_1,g_2)^\top$ and unobserved heterogeneity $(\varepsilon_1,\varepsilon_2)^\top$ are unaffected by the exogenous manipulation.
We call the copula $C_{Y_1^*Y_2^*}$ of $(Y^*_1,Y^*_2)^\top$ the \emph{counterfactual copula}, in contrast to the \emph{actual copula} $C_{Y_1Y_2}$, which is the copula of $(Y_1,Y_2)^\top$ in the status quo.
The policy effects of interest could be the difference in copula, $C_{Y_1^*Y_2^*}-C_{Y_1Y_2}$, and the difference in some measure of association, $\nu(C_{Y_1^*Y_2^*})-\nu(C_{Y_1Y_2})$, where $\nu$ is a functional defined on the collection of all copulas.
For example, the change in Spearman's rho between bivariate outcomes could be the main concern, as it corresponds to the rank-rank slope, a commonly used measure of intergenerational mobility in economics (\citeauthor{Chetty2014}, \citeyear{Chetty2014}).
In this case, the specific functional is $\nu_{\rho}:C\mapsto 12\int_{[0,1]^2} u_1u_2\d C(u_1,u_2)-3$.

We first show identification of the counterfactual copula under standard assumptions and adopt the principle of analogy to propose a nonparametric estimator $\widehat{C}_{Y_1^*Y_2^*}$ of the counterfactual copula in the presence of a random sample $\{(Y_{1i},Y_{2i},X_i,X_i^*)\}_{i=1}^n$.
This proposed estimator, together with the empirical copula estimator $\widehat{C}_{Y_1Y_2}$ of the actual copula, allows researchers to evaluate the effects of a policy before its implementation.
We then establish the joint weak convergence of $(\sqrt{n}(\widehat{C}_{Y_1^*Y_2^*}-C_{Y_1^*Y_2^*}),\sqrt{n}(\widehat{C}_{Y_1Y_2}-C_{Y_1Y_2}))^\top$ and $\sqrt{n}((\nu(\widehat{C}_{Y_1^*Y_2^*})-\nu(\widehat{C}_{Y_1Y_2}))-(\nu(C_{Y_1^*Y_2^*})-\nu(C_{Y_1Y_2})))$, provided that $\nu$ is Hadamard differentiable.
Since the limiting processes would have complicated covariance structures with unknown nuisance functions, we also validate the nonparametric bootstrap to draw inference on counterfactual quantities of interest.
Results of the simulation show that the counterfactual quantities of interest are $\sqrt{n}$-consistent, which is in line with our theoretical analysis.

Viewing our method as an essential complement to the existing toolkit for policy evaluation, we apply it to evaluating the effects of college education on intergenerational income mobility in the United States.
We use data obtained from the Panel Study of Income Dynamics (PSID) and consider two counterfactual scenarios.
In the first scenario, children with low educational attainment are assigned more years of schooling as if extra ``compulsory'' schooling years were required.
We find that conferring a college degree on all children can substantially reduce income persistence by about one-third of the magnitude of the considered association measure, including Spearman's rho, Kendall's tau, Gini's gamma, and Blomqvist's beta.
This finding echoes the sheepskin effects in literature, indicating that a college degree provides additional returns beyond educational attainment alone (\citeauthor{Hungerford1987}, \citeyear{Hungerford1987}; \citeauthor{Jaeger1996}, \citeyear{Jaeger1996}).
From a pragmatic perspective, a policy typically covers some, rather than all, families; additionally, existing empirical results (\citeauthor{Peters1992}, \citeyear{Peters1992}; \citeauthor{Sikhova0}, \citeyear{Sikhova0}) suggest that parental education is a key determinant of offspring income.
We thus consider the second scenario in which \emph{only} children from less educated families are required to complete a four-year college degree.
Our empirical results show that income persistence across generations is significantly lower if only children of parents with educational attainment less than, or equal to, ten years of schooling are required to complete a four-year college degree; in this case, families affected by this counterfactual policy account for about 15\% of the sample.
In contrast, expanding this policy to include families at the upper tail of the parental education distribution would result in only a negligible improvement in income mobility.

The rest of this paper is organized as follows. \cref{Model} discusses the model and relevant measures of association, and provides identification and estimation of the counterfactual copula. In \cref{Asymptotic}, we derive the joint weak convergence of the counterfactual and empirical copula processes and validate the nonparametric bootstrap for counterfactual quantities. \cref{Simulation} reports the simulation results, and \cref{EmpiricalStudy} presents the empirical study regarding the heterogeneous effects of college education on intergenerational mobility.
Finally, \cref{Conclusion} concludes.
Technical proofs of theorems and lemmas are deferred to \cref{Appendix}.

\section{Identification and Estimation}\label{Model}
As described in the introduction, we consider two continuous outcome variables
 \begin{align*}
Y_1=g_1(X,\varepsilon_1)\;\;\text{and}\;\;
Y_2=g_2(X,\varepsilon_2),
\end{align*}
where $X$ is a $d$-dimensional vector of covariates, $\varepsilon_1$ and $\varepsilon_2$ are individual unobserved heterogeneity, and $g_1$ and $ g_2$ are structural functions.\footnote{The nonseparable model could be viewed as the reduced-form equations induced by some underlying structural equations. A simple example is the triangular nonparametric simultaneous equations model
\begin{align*}
Y_1=\tilde{g}_1(Y_2)+e_1\;\;\text{and}\;\;
Y_2=\tilde{g}_2(X)+e_2.
\end{align*}
Substituting the second equation into the first yields the reduced-form equations
\begin{align*}
Y_1=\tilde{g}_1(\tilde{g}_2(X)+e_2)+e_1
=g_1(X,\varepsilon_1)\;\;\text{and}\;\;
Y_2=\tilde{g}_2(X)+e_2=g_2(X,\varepsilon_2),
\end{align*}
where $\varepsilon_1=(e_1,e_2)^\top$ and $\varepsilon_2=e_2$.
Unlike \citet{NeweyPowellEtAl1999}, who are concerned with the structural function $\tilde{g}_1$, we are interested in the dependence structure between $Y_1$ and $Y_2$.}
This bivariate setting is very flexible.
First, the structural functions $g_1$ and $ g_2$ are unknown to researchers, so \textit{ad hoc} parametric assumptions need not be imposed.
In addition, the policy effect could depend on unobserved heterogeneity because both $\varepsilon_1$ and $\varepsilon_2$ are nonseparable.
The dimensionality of $(\varepsilon_1,\varepsilon_2)^\top$ is also unspecified, so incorrect inference due to the exclusion of some important heterogeneity can be avoided.
Finally, it is possible that some elements of $X$ might be excluded from the structural functions.\footnote{
To see this, denoting $X=(X_0,X_1,X_2)^\top$, we write $Y_1=g_1(X_0,X_1,\varepsilon_1)$ and $Y_2=g_2(X_0,X_2,\varepsilon_2)$ in the case where $X_0$ affects $(Y_1,Y_2)^\top$, $X_1$ affects only $Y_1$, and $X_2$ affects only $Y_2$.}

Suppose that a counterfactual policy, manipulating the covariate from $X$ to $X^*$, causes a change in bivariate outcomes from $(Y_1,Y_2)^\top$ to $(Y^*_1,Y^*_2)^\top$, where
\begin{align*}
Y_1^*=g_1(X^*,\varepsilon_1)\;\;\text{and}\;\;
Y_2^*=g_2(X^*,\varepsilon_2).
\end{align*}
This change in bivariate outcomes in turn causes a change in their copula, and the difference between the counterfactual copula and the actual copula
\[
\Delta_{C}(u_1,u_2)\equiv C_{Y_1^*Y_2^*}(u_1,u_2)-C_{Y_1Y_2}(u_1,u_2)
\]
is called the \emph{policy effect on the copula}. The change in copula may be accompanied by changes in measures of association; hence, policymakers would be interested in the corresponding \emph{policy effect on the association}
\[
\Delta_{\nu}\equiv \nu(C_{Y_1^*Y_2^*})-\nu(C_{Y_1Y_2})
\]
for the following common measures:
\begin{enumerate}[(i)]
\item
Spearman's rho $\nu_{\rho}:C\mapsto 12\int_{[0,1]^2} u_1u_2\d C(u_1,u_2)-3$;
\item
Kendall's tau $\nu_{\tau}:C\mapsto 4\int_{[0,1]^2} C(u_1,u_2)\d C(u_1,u_2)-1$;
\item
Gini's gamma $\nu_{\gamma}:C\mapsto 2\int_{[0,1]^2} [|u_1+u_2-1|-|u_1-u_2|]\d C(u_1,u_2)$;
\item
Blomqvist's beta $\nu_{\beta}: C\mapsto 4C(0.5,0.5)-1$.
\end{enumerate}
We refer to \citet{Nelsen2006} for an excellent overview of these measures.

Since each aforementioned measure of association can be written in terms of some functional of the copula, the identification of those counterfactual policy effects on association is achieved if we can identify the counterfactual copula $C_{Y_1^*Y_2^*}$ and actual copula $C_{Y_1Y_2}$.
\citeauthor{Sklar1959}'s (\citeyear{Sklar1959}) theorem implies that the actual copula is identified by
\begin{align*}
C_{Y_1Y_2}(u_1,u_2)=F_{Y_1Y_2}(F^{-1}_{Y_1}(u_1),F^{-1}_{Y_2}(u_2)).
\end{align*}
Given independent and identically distributed observations $\{(Y_{1i},Y_{2i},X_i,X_i^*)\}_{i=1}^n$, the actual copula of $(Y_1,Y_2)^\top$ is commonly estimated by the empirical copula proposed by \cite{Deheuvels1979}; to be specific,
\begin{equation}
\label{empiricalcopula}
\widehat{C}_{Y_1Y_2}(u_1,u_2)\equiv\frac{1}{n}\sum_{i=1}^n\1\{Y_{1i}\leq \widehat{F}_{Y_1}^{-1}(u_1),Y_{2i}\leq \widehat{F}_{Y_2}^{-1}(u_2)\},
\end{equation}
where $\widehat F^{-1}_{Y_j}(u)=\inf\{y:\widehat F_{Y_j}(y)\geq u\}$ is the inverse function of the empirical CDF estimator $\widehat{F}_{Y_j}(y)=n^{-1}\sum_{i=1}^n\1\{Y_{ji}\leq y\}$ for $j\in\{1,2\}$.

To achieve the identification of counterfactual copula, we start by introducing assumptions as follows.\\[0.2cm]
\noindent
\textbf{Assumption I} (Identification)
\begin{enumerate}[label=I\arabic*]
\item \label{I1} $(X,X^*)^\top$ and $(\varepsilon_1,\varepsilon_2)^\top$ are independent.
\item \label{I2} The support of $X^*$ is a subset of the support of $X$.
\item \label{I3} The counterfactual outcome variables $(Y_1^*,Y_2^*)^\top$ are continuously distributed.
\end{enumerate}

\noindent
Assumption~\ref{I1} requires exogeneity of $(X,X^*)^\top$, which may be strong in some empirical studies, and is sufficient for Assumption 2.3 in \citet{Hsu2022}.\footnote{
Assumption~\ref{I1} could be relaxed by the control function approach but such consideration goes beyond the scope of the present paper.}
Assumption~\ref{I2} excludes the extrapolation of covariates.
Such extrapolation is inappropriate because none of parametric setups is imposed on the functions $g_1$ and $g_2$ outside the support of $X$.
Assumptions~\ref{I1} and~\ref{I2}, together with the availability of data on $(Y_1,Y_2,X,X^*)^\top$, allow us to identify the joint counterfactual cumulative distribution function (CDF) $F_{Y_1^*Y_2^*}$ and marginal counterfactual CDFs $(F_{Y_1^*},F_{Y_2^*})^\top$,
as shown in \cite{Rothe2010} and \cite{ChernozhukovFernandez-ValEtAl2013a}.
Moreover, Assumption~\ref{I3} allows us to uniquely express the counterfactual copula $C_{Y_1^*Y_2^*}$ in terms of counterfactual CDFs $(F_{Y_1^*Y_2^*},F_{Y_1^*},F_{Y_2^*})^\top$ by Sklar's theorem.
More details of Sklar's theorem can be found in \cite{Nelsen2006}.
We summarize the discussion in the following proposition.

\begin{pro}
\label{PRO1}
Under Assumption {\normalfont I}, the counterfactual copula is identified by
\begin{align*}
C_{Y_1^*Y_2^*}(u_1,u_2)
&=F_{Y_1^*Y_2^*}(F^{-1}_{Y_1^*}(u_1),F^{-1}_{Y_2^*}(u_2))\\
&=\int F_{Y_1Y_2|X}(F^{-1}_{Y_1^*}(u_1),F^{-1}_{Y_2^*}(u_2)|x)F_{X^*}(\d x),
\end{align*}
where $F^{-1}_{Y_j^*}(u)\equiv\inf\{y:F_{Y_j^*}(y)\geq u\}$ is the inverse function of
\[
F_{Y_j^*}(y)=\int F_{Y_j|X}(y|x)F_{X^*}(\d x)
\]
for every $j\in\{1,2\}$, and
\[
F_{Y_{1}^{*}Y_{2}^{*}}(y_{1},y_{2})=\int F_{Y_{1}Y_{2}|X}(y_{1},y_{2}|x)\; F_{X^{*}}(\d x).
\]
\end{pro}

\cref{PRO1} shows that the counterfactual copula $C_{Y_1^*Y_2^*}$ is identified without any restriction on the correlation between $\varepsilon_1$ and $\varepsilon_2$.\footnote{If $\varepsilon_1$ and $\varepsilon_2$ are conditionally independent given $X$, then the counterfactual copula can be further simplified to
\[
C_{Y_1^*Y_2^*}(u_1,u_2)
=\int F_{Y_1|X}(F^{-1}_{Y_1^*}(u_1)|x)F_{Y_2|X}(F^{-1}_{Y_2^*}(u_2)|x)F_{X^*}(\d x).
\]
In addition, the formula of $F_{Y_j^*}$ can be simplified if not all elements of $X$ are factors of $Y_j$ for $j\in\{1,2\}$.}
This proposition also suggests a nonparametric plug-in counterfactual copula estimator
\begin{align}
\label{copulaestimator}
\widehat{C}_{Y_1^*Y_2^*}(u_1,u_2)\equiv\frac{1}{n}
\sum_{i=1}^n\widehat{F}_{Y_1Y_2|X}(\widehat{F}^{-1}_{Y_1^*}(u_1),\widehat{F}^{-1}_{Y_2^*}(u_2)|X_i^*),
\end{align}
where for a kernel function $K$ and a bandwidth $h=h_n$ shrinking with the sample size $n$,
\[
\widehat{F}_{Y_1Y_2|X}(y_1,y_2|x)=\frac{\sum_{i=1}^n\1\{Y_{1i}\leq y_1,Y_{2i}\leq y_2\}K\left(\frac{X_i-x}{h}\right)}{\sum_{i=1}^nK\left(\frac{X_i-x}{h}\right)},
\]
and for $j\in\{1,2\}$, $\widehat F^{-1}_{Y_j^*}(u)=\inf\{y:\widehat F_{Y_j^*}(y)\geq u\}$ is the inverse function of
\[
\widehat{F}_{Y_j^*}(y)=\frac{1}{n}\sum_{i=1}^n\widehat{F}_{Y_j|X}(y|X_i^*)
\]
with $\widehat{F}_{Y_1|X}(y|x)=\widehat{F}_{Y_1Y_2|X}(y,\infty|x)$ and $\widehat{F}_{Y_2|X}(y|x)=\widehat{F}_{Y_1Y_2|X}(\infty,y|x)$.

The proposed estimator in \cref{copulaestimator} is easy to implement.
To see this, we define the joint counterfactual CDF estimator to be
\begin{align}
\label{CDFestimator}
\widehat{F}_{Y_{1}^{*}Y_{2}^{*}}(y_{1},y_{2})
\equiv\frac{1}{n}\sum_{i=1}^{n}W_{i,n}\1\{Y_{1i}\leq y_{1},Y_{2i}\leq y_{2}\},
\end{align}
where the weights $\{W_{i,n}\}_{i=1}^{n}$ are given by
\begin{align*}
W_{i,n}=\sum_{j=1}^n\frac{K\left(\frac{X_i-X^*_j}{h}\right)}{\sum_{\ell=1}^nK\left(\frac{X_{\ell}-X^*_j}{h}\right)}.
\end{align*}
The joint counterfactual estimator in \cref{CDFestimator} can be viewed as an extension of \citeauthor{Rothe2010}'s (\citeyear{Rothe2010}) marginal counterfactual CDF estimator to the two-dimensional setting; specifically, in this study, we have $\widehat{F}_{Y_1^*}(y_{1})=\widehat{F}_{Y_{1}^{*}Y_{2}^{*}}(y_{1},\infty)$ and $\widehat{F}_{Y_2^*}(y_{2})=\widehat{F}_{Y_{1}^{*}Y_{2}^{*}}(\infty,y_{2})$.
The estimator $\widehat{F}_{Y_{1}^{*}Y_{2}^{*}}$, compared with $\widehat{F}_{Y_{1}Y_{2}}$ can be also applied to testing bivariate stochastic dominance, which is an important topic in decision making; see for example \citet{DenuitHuangEtAl2014}.
More importantly, we can rewrite the counterfactual copula estimator in \cref{copulaestimator} as
\begin{align*}
\widehat{C}_{Y_1^*Y_2^*}(u_1,u_2)
=\frac{1}{n}\sum_{i=1}^nW_{i,n}\1\{Y_{1i}\leq \widehat{F}_{Y^*_1}^{-1}(u_1),Y_{2i}\leq
\widehat{F}_{Y^*_2}^{-1}(u_2)\}.
\end{align*}
This formula suggests that when evaluating $\widehat{C}_{Y_{1}^{*}Y_{2}^{*}}$ at multiple locations, we only require a one-off computation of weights $\{W_{i,n}\}_{i=1}^{n}$.\footnote{
Note that whenever every weight $W_{i,n}$ is nonnegative, both counterfactual CDF estimators $\widehat{F}_{Y^*_1}$ and $\widehat{F}_{Y^*_2}$ are nondecreasing; in this case, the counterfactual copula estimator $\widehat{C}_{Y_1^*Y_2^*}$ is $2$-increasing because for any $(u_1,u_2), (v_1,v_2)\in [0,1]^2$ with $u_1\leq v_1$ and $u_2\leq v_2$,
\begin{align*}
&\widehat{C}_{Y_1^*Y_2^*}(v_1,v_2)
-\widehat{C}_{Y_1^*Y_2^*}(v_1,u_2)
-\widehat{C}_{Y_1^*Y_2^*}(u_1,v_2)
+\widehat{C}_{Y_1^*Y_2^*}(u_1,u_2)\\
={}&\frac{1}{n}\sum_{i=1}^nW_{i,n}
\1\{\widehat{F}_{Y^*_1}^{-1}(u_1)<Y_{1i}\leq\widehat{F}_{Y^*_1}^{-1}(v_1)\}
\1\{\widehat{F}_{Y^*_2}^{-1}(u_2)<Y_{2i}\leq\widehat{F}_{Y^*_2}^{-1}(v_2)\}
\geq 0.
\end{align*}
In the case where some weights are negative, we are unaware of any study dealing with necessary adjustments to satisfy the $2$-increasing property of a bivariate CDF.
This task would be challenging and reserved for future work.
}
Although the counterfactual copula estimator in \cref{copulaestimator} has the same weights $\{W_{i,n}\}_{i=1}^n$ as those in \cite{Rothe2010},
it has summands with an indicator function evaluated at two estimated arguments $(\widehat{F}_{Y^*_1}^{-1}(u_1),\widehat{F}_{Y^*_2}^{-1}(u_2))^\top$, which introduces nonnegligible estimation effects, as we will see in the next section.

Combining estimators in \cref{copulaestimator,empiricalcopula}, we can estimate the policy effect on the copula by
\[
\widehat{\Delta}_{C}(u_1,u_2)\equiv\widehat{C}_{Y_1^*Y_2^*}(u_1,u_2)-\widehat{C}_{Y_1Y_2}(u_1,u_2),
\]
and the policy effect on the association by
\[
\widehat{\Delta}_{\nu}\equiv\nu(\widehat{C}_{Y_1^*Y_2^*})-\nu(\widehat{C}_{Y_1Y_2}).
\]

\section{Asymptotic Theory}
\label{Asymptotic}
The asymptotic properties of the empirical copula estimator in \cref{empiricalcopula} has been considerably studied over the past two decades.
One recent advance is that \cite{Segers2012} proves the weak convergence of the \emph{empirical copula process} $\widehat{\mathbb{C}}_{Y_1Y_2}\equiv\sqrt{n}(\widehat{C}_{Y_1Y_2}-C_{Y_1Y_2})$ in the space $(\ell^\infty([0,1]^2),\|\cdot\|_\infty)$ under the non-restrictive assumption on the copula in cases where data are independent and identically distributed.\footnote{
We write $(\ell^\infty([0,1]^2),\|\cdot\|_\infty)$ for the set of all bounded real-valued functions on $[0,1]^2$ equipped with the uniform norm $\|\cdot\|_{\infty}$.}
We present this assumption on the actual copula as follows.\\[0.2cm]
\textbf{Assumption C} (Non-restrictive Assumption on the Actual Copula)
\begin{enumerate}[label=C\arabic*]
\item[] \label{C} For each $j\in\{1,2\}$, the $j$-th partial derivative $\partial C_{Y_1Y_2}/\partial u_j$ exists and is continuous on the set $\{(u_1,u_2)^\top\in[0,1]^2:u_j\in(0,1)\}$.
\end{enumerate}

To extend the partial derivatives to the whole unit square $[0,1]^2$, we follow \cite{Segers2012} and define
\[
\frac{\partial C_{Y_1Y_2}(u_1,u_2)}{\partial u_1}
\equiv\begin{cases}
\limsup\limits_{u\downarrow0}\dfrac{C_{Y_1Y_2}(u,u_2)}{u} & \text{if $u_1=0$}, \\
\limsup\limits_{u\downarrow0}\dfrac{C_{Y_1Y_2}(1,u_2)-C_{Y_1Y_2}(1-u,u_2)}{u} & \text{if $u_1=1$},
\end{cases}
\]
and
\[
\frac{\partial C_{Y_1Y_2}(u_1,u_2)}{\partial u_2}
\equiv\begin{cases}
\limsup\limits_{u\downarrow0}\dfrac{C_{Y_1Y_2}(u_1,u)}{u} & \text{if $u_2=0$}, \\
\limsup\limits_{u\downarrow0}\dfrac{C_{Y_1Y_2}(u_1,1)-C_{Y_1Y_2}(u_1,1-u)}{u} & \text{if $u_2=1$}.
\end{cases}
\]
As discussed by \citet{Segers2012} and \citet{Buecher2014}, Assumption C is necessary for the existence of the limiting process and continuity of its trajectories.
Furthermore, \cite{Buecher2013} demonstrate that Assumption C implies Hadamard differentiability of the copula mapping; thus, weak convergence of the empirical copula process can be established as long as the empirical process $\sqrt{n}(\widetilde{G}_{Y_1Y_2}-C_{Y_1Y_2})$ converges weakly in the space $(\ell^\infty([0,1]^2),\|\cdot\|_\infty)$ to a tight, centered Gaussian process, where the empirical CDF
\[
\widetilde{G}_{Y_1Y_2}(u_1,u_2)\equiv\frac{1}{n}\sum_{i=1}^n\1\{F_{Y_1}(Y_{1i})\leq u_1,F_{Y_2}(Y_{2i})\leq u_2\}
\]
is constructed by the unobservable probability integral transforms $\{(F_{Y_1}(Y_{1i}),F_{Y_2}(Y_{2i}))\}_{i=1}^n$.
For completeness,
we summarize the discussion above and state the following proposition, which is obtained by \cite{Segers2012} for independent observations.\footnote{\citet{Buecher2013} and \citet{BuecherKojadinovic2016} obtain the same representation of $\widehat{\mathbb{C}}_{Y_1Y_2}$ for mixing observations.}

\begin{pro}
\label{PRO2}
Under Assumption {\normalfont C}, the empirical copula process has the representation
\begin{align*}
&\widehat{\mathbb{C}}_{Y_1Y_2}(u_1,u_2)\\
={}&\sqrt{n}(\widetilde{G}_{Y_1Y_2}(u_1,u_2)-C_{Y_1Y_2}(u_1,u_2))
-\sum_{j=1}^2\frac{\partial C_{Y_1Y_2}(u_1,u_2)}{\partial u_j}\sqrt{n}(\widetilde{G}_{Y_j}(u_j)-u_j)+R_n(u_1,u_2),
\end{align*}
where
$\widetilde{G}_{Y_1}(u)=\widetilde{G}_{Y_1Y_2}(u,1)$,
$\widetilde{G}_{Y_2}(u)=\widetilde{G}_{Y_1Y_2}(1,u)$,
and the remainder term $R_n$ satisfies
\[
\sup_{(u_1,u_2)\in[0,1]^2}|R_n(u_1,u_2)|=o_p(1).
\]
Consequently, the empirical copula process converges weakly to a Gaussian process in the space $(\ell^\infty([0,1]^2),\|\cdot\|_\infty)$.
\end{pro}

We now show that the \emph{counterfactual copula process} $\widehat{\mathbb{C}}_{Y^*_1Y^*_2}\equiv\sqrt{n}(\widehat{C}_{Y^*_1Y^*_2}-C_{Y^*_1Y^*_2})$ has a similar representation provided the following assumptions hold.\\[0.3cm]
\noindent
\textbf{Assumption K} (Kernel)\\
The kernel function $K:\mathbb{R}^{d}\rightarrow \mathbb{R}$ satisfies the following conditions.
\begin{enumerate}[label=K\arabic*]
\item \label{K1} $K$ is of bounded variation, vanishes outside $[-1,1]^d$, and satisfies $\int K(u)\d u=1$.
\item \label{K2} There is a positive integer $r\geq 2$ such that $\int\prod_{\ell=1}^du_{\ell}^{\lambda_{\ell}}K(u)\d u=0$ for any $d$-dimensional vector $\lambda=(\lambda_1,\ldots,\lambda_{d})^\top$ of nonnegative integers with $\sum_{\ell=1}^d\lambda_{\ell}\leq r-1$.
\item \label{K3} For $u\in[-1,1]^d$, $K(u)$ is $r$-times differentiable with respect to $u$, and the derivatives are uniformly continuous and bounded.
\item \label{K4} For $u\in[-1,1]^d$, $K(u)=K(|u|)$.
\end{enumerate}
\textbf{Assumption B} (Bandwidth)\\
The sequence $\{h_{n}\}_{n=1}^{\infty}$ of bandwidths satisfies the following conditions.
\begin{enumerate}[label=B\arabic*]
\item \label{B1} $h_{n}\rightarrow 0$.
\item \label{B2} $\frac{\log n}{n^{1/2}h_{n}^{d}}\rightarrow 0$.
\item \label{B3} $n^{1/2}h_{n}^{r}\rightarrow 0$.
\end{enumerate}

\noindent
These assumptions on the kernel and bandwidth are mild and common in the literature on nonparametric and semiparametric econometrics.
To make Assumptions~\ref{B2} and~\ref{B3} valid simultaneously, we require that the order of kernel $K$ be greater than the dimension of covariates, as is imposed in \citet{Rothe2010}.
We also need technical assumptions on the support and smoothness of density and distribution functions as follows.\\[0.3cm]
\noindent
\textbf{Assumption S} (Smoothness and Support)\\
Let $\X$, $\Y_1$, $\Y_2$, $\Y_1^*$, and $\Y_2^*$ be the support of $X$, $Y_1$, $Y_2$, $Y_1^*$, and $Y_2^*$, respectively. Also let $f_Z$ denote the density function of a generic random variable $Z$ with respect to the Lebesgue measure.
\begin{enumerate}[label=S\arabic*]
\item \label{S1}
The function $f_{X}(x)$ is $r$-times differentiable with respect to $x$ on the interior of $\X$, and its derivatives are bounded and uniformly continuous.
In addition, the function $f_{X}(x)$ is bounded away from zero on $\X$.
\item \label{S2}
For all $(y_1,y_2,x)\in \Y_1\Y_2\X$, the first $r$ partial derivatives with respect to $x$ of $F_{Y_1Y_2|X}(y_1,y_2|x)$ and $f_{Y_1Y_2|X}(y_1,y_2|x)$ are bounded.
\item \label{S3}
The function $f_{X^*}(x)$ is $r$-times differentiable with respect to $x$ on the interior of $\X$, and its derivatives are bounded and uniformly continuous.\footnote{Let $f_{X^*}(x)=0$ if $x\notin\X^*$.}
\item \label{S4}
Both $\E[(\sup_{(y_1,y_2)\in\Y_1\Y_2}f_{Y_1Y_2|X}(y_1,y_2|X))^2]$ and $\E[ (f_{X^*}(X)/f_{X}(X))^2]$ are finite.
\item \label{S5}
For each $j\in\{1,2\}$, both $\Y_j$ and $\Y^*_j$ are bounded with respect to the Euclidean distance.
\item \label{S6}
The function $F_{Y^*_1Y^*_2}$ has continuous partial derivatives of order up to 2 on $\Y_1^*\Y_2^*$.
\item \label{S7}
The function $f_{Y^*_1Y^*_2}$ is bounded away from zero on $\Y_1^*\Y_2^*$.
\end{enumerate}

\begin{pro}
\label{PRO3}
Under Assumptions {\normalfont I}, {\normalfont K}, {\normalfont B}, and {\normalfont S}, the counterfactual copula process has the representation
\begin{align*}
&\widehat{\mathbb{C}}_{Y^*_1Y^*_2}(u_1,u_2)\\
={}&\sqrt{n}(\widetilde{G}_{Y^*_1Y^*_2}(u_1,u_2)-C_{Y^*_1Y^*_2}(u_1,u_2))-\sum_{j=1}^2\frac{\partial C_{Y^*_1Y^*_2}(u_1,u_2)}{\partial u_j}\sqrt{n}(\widetilde{G}_{Y^*_j}(u_j)-u_j)+R^*_n(u_1,u_2),
\end{align*}
where $\widetilde{G}_{Y^*_1}(u)=\widetilde{G}_{Y^*_1Y^*_2}(u,1)$ and
$\widetilde{G}_{Y^*_2}(u)=\widetilde{G}_{Y^*_1Y^*_2}(1,u)$ with
\[
\widetilde{G}_{Y^*_1Y^*_2}(u_1,u_2)\equiv\frac{1}{n}\sum_{i=1}^nW_{i,n}\1\{F_{Y^*_1}(Y_{1i})\leq u_1,F_{Y^*_2}(Y_{2i})\leq u_2\},
\]
and $\{W_{i,n}\}_{i=1}^n$ are the weights given immediately after \cref{CDFestimator}. The extension to $[0,1]^2$ of $\partial C_{Y^*_1Y^*_2}/\partial u_j$ is similar to that of $\partial C_{Y_1Y_2}/\partial u_j$ for $j\in\{1,2\}$ , and the remainder term $R^*_n$ satisfies
\[
\sup_{(u_1,u_2)\in[0,1]^2}|R^*_n(u_1,u_2)|=o_p(1).
\]
\end{pro}

According to \cref{PRO3}, the weak convergence of $\widehat{\mathbb{C}}_{Y^*_1Y^*_2}$ in $(\ell^\infty([0,1]^2),\|\cdot\|_\infty)$ is guaranteed by that of $\sqrt{n}(\widetilde{G}_{Y^*_1Y^*_2}-C_{Y^*_1Y^*_2})$ in the same space. Notice that $\widetilde{G}_{Y^*_1Y^*_2}(u_1,u_2)
=\widehat{F}_{Y_1^*Y_2^*}(F_{Y_1^*}^{-1}(u_1),F_{Y_2^*}^{-1}(u_2))$ by construction, where $\widehat{F}_{Y_1^*Y_2^*}$ is the joint counterfactual CDF estimator defined in \cref{CDFestimator}.
It follows that
\begin{align*}
&\sqrt{n}(\widetilde{G}_{Y^*_1Y^*_2}(u_1,u_2)-C_{Y^*_1Y^*_2}(u_1,u_2))\\
={}&\sqrt{n}(\widehat{F}_{Y_1^*Y_2^*}(F_{Y_1^*}^{-1}(u_1),F_{Y_2^*}^{-1}(u_2))-F_{Y_1^*Y_2^*}(F_{Y_1^*}^{-1}(u_1),F_{Y_2^*}^{-1}(u_2))).
\end{align*}
Hence, the weak convergence of $\sqrt{n}(\widetilde{G}_{Y^*_1Y^*_2}-C_{Y^*_1Y^*_2})$ in $(\ell^\infty([0,1]^2),\|\cdot\|_\infty)$ is equivalent to that of $\sqrt{n}(\widehat{F}_{Y_1^*Y_2^*}-F_{Y_1^*Y_2^*})$ in $(\ell^\infty(\Y_1\Y_2),\|\cdot\|_\infty)$. To show the latter weak convergence, we extend the representation of the univariate counterfactual CDF estimator $\widehat{F}_{Y_j^*}$ obtained in \cite{Rothe2010} to that of joint counterfactual CDF estimator $\widehat{F}_{Y_1^*Y_2^*}$ in \cref{LEM1} in \cref{Appendix}.

Since the policy effects on the copula and the association measures are based on comparing the actual and counterfactual copulas, we study the asymptotic behavior of the random map
\[
\begin{bmatrix}
u_1\\
u_2
\end{bmatrix}
\mapsto
\begin{bmatrix}
\sqrt{n}(\widehat{C}_{Y_1Y_2}(u_1,u_2)-C_{Y_1Y_2}(u_1,u_2))\\
\sqrt{n}(\widehat{C}_{Y^*_1Y^*_2}(u_1,u_2)-C_{Y^*_1Y^*_2}(u_1,u_2))
\end{bmatrix}
=
\begin{bmatrix}
\widehat{\mathbb{C}}_{Y_1Y_2}(u_1,u_2)\\
\widehat{\mathbb{C}}_{Y^*_1Y^*_2}(u_1,u_2)
\end{bmatrix}.
\]
\cref{PRO2,PRO3} allow the random map to be approximated by the scaled sum of two-dimensional influence functions, which will be used to prove \cref{THM1} below.\footnote{The notation $\Rightarrow$ stands for weak convergence in a functional space equipped with an appropriate norm.}

\begin{thm}
\label{THM1}
Suppose that Assumptions {\normalfont I}, {\normalfont C}, {\normalfont K}, {\normalfont B}, and {\normalfont S} hold.
Then in the space $(\ell^\infty([0,1]^2),\|\cdot\|_\infty)\times(\ell^\infty([0,1]^2),\|\cdot\|_\infty)$,
\[
(\widehat{\mathbb{C}}_{Y_1Y_2},\widehat{\mathbb{C}}_{Y^*_1Y^*_2})^\top
\Rightarrow (\mathbb{C}_{Y_1Y_2},\mathbb{C}_{Y^*_1Y^*_2})^\top\equiv\mathbb{C},
\]
where $\mathbb{C}$ is a two-dimensional centered Gaussian process.

Furthermore, let $\nu$ be a functional mapping from $(\ell^\infty([0,1]^2),\|\cdot\|_\infty)$ to some normed space $(\mathcal{V},\|\cdot\|_{\mathcal{V}})$.
Suppose that $\nu$ is Hadamard differentiable at $C_{Y_1Y_2}$ and $C_{Y^*_1Y^*_2}$.
Then in the space $(\mathcal{V},\|\cdot\|_{\mathcal{V}})\times(\mathcal{V},\|\cdot\|_{\mathcal{V}})$,
\begin{align*}
(\sqrt{n}(\nu(\widehat{C}_{Y_1Y_2})-\nu(C_{Y_1Y_2})),
\sqrt{n}(\nu(\widehat{C}_{Y^*_1Y^*_2})-\nu(C_{Y^*_1Y^*_2})))^\top
\Rightarrow(\mathbb{V}_{Y_1Y_2},\mathbb{V}_{Y^*_1Y^*_2})^\top\equiv\mathbb{V},
\end{align*}
where $\mathbb{V}$ is a two-dimensional centered Gaussian process.
\end{thm}

\noindent
Applying \cref{THM1} and the continuous mapping theorem, we obtain $\sqrt{n}(\widehat{\Delta}_{C}-\Delta_{C})\Rightarrow (-1,1)\mathbb{C}$ in the space $(\ell^\infty([0,1]^2),\|\cdot\|_\infty)$ and $\sqrt{n}(\widehat{\Delta}_{\nu}-\Delta_{\nu})\Rightarrow (-1,1)\mathbb{V}$ in the space $(\mathcal{V},\|\cdot\|_{\mathcal{V}})$.

To find confidence bands for $C_{Y^*_1Y^*_2}$, $\Delta_{C}$, and $\Delta_{\nu}$,
we apply a \emph{nonparametric bootstrap} procedure.
Let $\{(Y_{b,1i},Y_{b,2i},X_{b,i},X_{b,i}^*)\}_{i=1}^n$ be the bootstrap data with sample size $n$ drawn with replacement from the data $\{(Y_{1i},Y_{2i},X_i,X_i^*)\}_{i=1}^n$ at hand. By replacing the data at hand with the bootstrap data in \cref{copulaestimator,empiricalcopula}, we construct the bootstrap counterfactual copula estimator $\widehat{C}^{b}_{Y_1^*Y_2^*}$ and bootstrap empirical copula estimator $\widehat{C}^{b}_{Y_1Y_2}$, respectively. Additionally, we define the \emph{bootstrap counterfactual copula process} to be $\widehat{\mathbb{C}}^{b}_{Y^*_1Y^*_2}
\equiv\sqrt{n}(\widehat{C}^{b}_{Y^*_1Y^*_2}-\widehat{C}_{Y^*_1Y^*_2})$
and \emph{bootstrap empirical copula process} $\widehat{\mathbb{C}}^{b}_{Y_1Y_2}\equiv\sqrt{n}(\widehat{C}^{b}_{Y_1Y_2}-\widehat{C}_{Y_1Y_2})$. To validate the nonparametric bootstrap for statistical inference, we prove the following theorem.

\begin{thm}
\label{THM2}
Under the assumptions in \cref{THM1}, in the space $(\ell^\infty([0,1]^2),\|\cdot\|_\infty)\times(\ell^\infty([0,1]^2),\|\cdot\|_\infty)$,
\[
(\widehat{\mathbb{C}}^{b}_{Y_1Y_2},\widehat{\mathbb{C}}^{b}_{Y^*_1Y^*_2})^\top\Rightarrow(\mathbb{C}_{Y_1Y_2},\mathbb{C}_{Y^*_1Y^*_2})^\top=\mathbb{C},
\]
conditional on the data $\{(Y_{1i},Y_{2i},X_i,X_i^*)\}_{i=1}^n$.

Moreover, if the functional mapping $\nu$ from $(\ell^\infty([0,1]^2),\|\cdot\|_\infty)$ to some normed space $(\mathcal{V},\|\cdot\|_{\mathcal{V}})$ is Hadamard differentiable at $C_{Y_1Y_2}$ and $C_{Y^*_1Y^*_2}$, then in the space $(\mathcal{V},\|\cdot\|_{\mathcal{V}})\times (\mathcal{V},\|\cdot\|_{\mathcal{V}})$,
\[
(\sqrt{n}(\nu(\widehat{C}^{b}_{Y_1Y_2})-\nu(\widehat{C}_{Y_1Y_2})),
\sqrt{n}(\nu(\widehat{C}^{b}_{Y^*_1Y^*_2})-\nu(\widehat{C}_{Y^*_1Y^*_2})))^\top\Rightarrow(\mathbb{V}_{Y_1Y_2},\mathbb{V}_{Y^*_1Y^*_2})^\top=\mathbb{V},
\]
conditional on the data $\{(Y_{1i},Y_{2i},X_i,X_i^*)\}_{i=1}^n$.
\end{thm}

\cref{THM2} allows us to construct confidence bands for the counterfactual copula $C_{Y_1^*Y_2^*}$ and the actual copulas $C_{Y_1Y_2}$.
It can also be applied to the construction of confidence intervals for an association measure $\nu(C_{Y^{*}_{1}Y^{*}_{2}})$ and a policy effect $\Delta_{\nu}= \nu(C_{Y_1^*Y_2^*})-\nu(C_{Y_1Y_2})$ for any functional mapping $\nu:(\ell^\infty([0,1]^2),\|\cdot\|_\infty)\to(\R,|\cdot|)$ that is Hadamard differentiable at $C_{Y_1Y_2}$ and $C_{Y^*_1Y^*_2}$. Details about the construction are deferred to the next section.

\section{Monte Carlo Simulation}
\label{Simulation}
We evaluate the finite-sample properties of the actual copula and counterfactual copula estimators.
We also investigate bootstrap validity for inference on the association measures and corresponding policy effects.
The data generating process under consideration is
\begin{align*}
Y_1=3+2X-\varepsilon_1\;\;\text{and}\;\;
Y_2=1+3X+2\varepsilon_2,
\end{align*}
where $(X,\varepsilon_1,\varepsilon_2)^\top$ follows a trivariate normal distribution with zero mean and covariance matrix
\[
\begin{bmatrix}
1	&	0	&	0	\\
0	&	1	&	0.5	\\
0	&	0.5	&	1	\\
\end{bmatrix}.
\]
Clearly, $(Y_1,Y_2)^\top$ is bivariate normally distributed and the actual copula is a Gaussian copula $C_{Y_1Y_2}(u_1,u_2)=\Phi_R(\Phi^{-1}(u_1),\Phi^{-1}(u_2))$, where $\Phi^{-1}$ is the inverse function of the standard normal CDF and $\Phi_R$ is the joint CDF of a bivariate normal with zero mean and correlation matrix
\[
R=
\begin{bmatrix}
1	&	\frac{\sqrt{65}}{13}	\\
\frac{\sqrt{65}}{13}	&	1	\\
\end{bmatrix}.
\]
We consider an exogenous policy intervention
\[
X^*=\pi(X)=0.5X,
\]
and this intervention does not affect the unobservables $(\varepsilon_1,\varepsilon_2)^\top$.
Simple algebra shows that the counterfactual copula is also a Gaussian copula $C_{Y_1^*Y_2^*}(u_1,u_2)=\Phi_{R^*}(\Phi^{-1}(u_1),\Phi^{-1}(u_2))$ with correlation matrix
\[
R^*=
\begin{bmatrix}
1	&	\frac{\sqrt{2}}{10}	\\
\frac{\sqrt{2}}{10}	&	1	\\
\end{bmatrix}.
\]
The data generating process and the counterfactual policy are not meant to mimic any data set in empirical studies; instead, they are only illustrative of the proposed method. 

Given a random sample $\{(Y_{1i},Y_{2i},X_i)\}_{i=1}^n$ of size $n$, we estimate the actual copula by the rank-based empirical copula in \citet{Genest1995} for computational efficiency; specifically,
\[
\widetilde{C}_{Y_1Y_2}(u_1,u_2)=\frac{1}{n}\sum_{i=1}^n\1\{\widehat{F}_{Y_1}(Y_{1i})\leq u_1,\widehat{F}_{Y_2}(Y_{2i})\leq u_2\}.
\]
This rank-based empirical copula $\widetilde{C}_{Y_1Y_2}$ differs from $\widehat{C}_{Y_1Y_2}$ in \cref{empiricalcopula} but their difference is first-order asymptotically negligible; see \citet{FermanianRadulovicEtAl2004}.
Similarly, we modify the proposed counterfactual copula estimator $\widehat{C}_{Y_1^*Y_2^*}$ in \cref{copulaestimator} as
\[
\widetilde{C}_{Y_1^*Y_2^*}(u_1,u_2)=\frac{1}{n}\sum_{i=1}^nW_{i,n}\1\{\widehat{F}_{Y_1^*}(Y_{1i})\leq u_1,\widehat{F}_{Y_2^*}(Y_{2i})\leq u_2\}.
\]
To construct the weights $W_{i,n}$, we use the Epanechnikov kernel with bandwidth $h=5.5s_Xn^{-1/3}$ where $s_X$ denotes the sample standard deviation of $X$.
We have also tried other kernel functions and bandwidth constants, but the results are not too sensitive to these changes.\footnote{In particular, we have tried the equivalent kernel of the local linear regression for automatic boundary correction. See \cite{Fan1996} for more details. The simulation results are broadly similar and are available upon request.}
The copula functions are estimated on the grids $\{(i/100, j/100): i, j= 0,1,\dots,100\}$.
For every sample size $n\in\{100, 200, 400\}$, we conduct 1,000 simulation replications.

\begin{table}[t]
\begin{threeparttable}
\centering
\caption{Simulation results for actual and counterfactual copula estimators.}
\label{table:copula}
\begin{tabular*}{\textwidth}{@{\extracolsep{\fill}}ccccccccc}
\toprule
MIAE	&	\multicolumn{1}{c}{Actual copula}	&	\multicolumn{2}{c}{Counterfactual copula}	\\
\cmidrule(lr){2-2}	\cmidrule(lr){3-4}
\multicolumn{1}{c}{$n$}	&	\multicolumn{1}{c}{Empirical estimator}	&	\multicolumn{1}{c}{Proposed estimator}	&	\multicolumn{1}{c}{Oracle estimator}	\\
\midrule
100		&	0.937	&	1.526	&	1.241	\\
200		&	0.657	&	1.097	&	0.866	\\
400		&	0.465	&	0.788	&	0.611	\\
\midrule
RMISE	&	\multicolumn{1}{c}{Actual copula}	&	\multicolumn{2}{c}{Counterfactual copula}	\\
\cmidrule(lr){2-2}	\cmidrule(lr){3-4}
\multicolumn{1}{c}{$n$}	&	\multicolumn{1}{c}{Empirical estimator}	&	\multicolumn{1}{c}{Proposed estimator}	&	\multicolumn{1}{c}{Oracle estimator}	\\
\midrule
100		&	1.364	&	2.090	&	1.674	\\
200		&	0.959	&	1.523	&	1.173	\\
400		&	0.678	&	1.100	&	0.829	\\
\bottomrule
\end{tabular*}
\begin{tablenotes}
\item \scriptsize{Note:
The results are obtained from 1,000 simulations in which the copula functions are estimated on the grids $\{(i/100, j/100): i, j= 0,1,\dots,100\}$. The figures are multiplied by 100 to improve readability. See texts for more estimation details.}
\end{tablenotes}
\end{threeparttable}
\end{table}

\cref{table:copula} presents the simulation results regarding the actual and counterfactual copula estimators in terms of mean integrated absolute error (MIAE) and root mean integrated squared error (RMISE).\footnote{For an estimator $\widehat{C}$ of a copula $C$, the mean integrated absolute error of $\widehat{C}$ is
\[
\text{MIAE}(\widehat{C})=\mathbb{E}\left[\int_0^1\int_0^1|\widehat{C}(u_1,u_2)-C(u_1,u_2)|\d u_1 \d u_2\right]
\]
and the root mean integrated squared error of $\widehat{C}$ is
\[
\text{RMISE}(\widehat{C})=\sqrt{\mathbb{E}\left[\int_0^1\int_0^1|\widehat{C}(u_1,u_2)-C(u_1,u_2)|^2\d u_1\d u_2\right]}.
\]}
As expected, the rank-based empirical copula $\widetilde{C}_{Y_1Y_2}$ performs well in terms of MIAE and RMISE.
It can also be seen that the MIAE and RMISE of $\widetilde{C}_{Y_1^*Y_2^*}$ shrink with increasing sample size and nearly halve as the size quadruples.
This result suggests the $\sqrt{n}$-consistency of the proposed counterfactual estimator, which is in line with our theoretical analysis.
For comparison, we also consider an \emph{oracle} counterfactual copula estimator, which would be the empirical copula of $(Y_1^*,Y_2^*)$ if the counterfactual outcomes $\{(Y_{1i}^*,Y_{2i}^*)\}_{i=1}^n$ were observed.
This oracle estimator, despite having smaller MIAE and RMISE, does not outweigh the proposed estimator considerably for moderate sample sizes.

We next turn to address the inferential questions about the measures of association and corresponding policy effects.
To tackle the inferential tasks, we implement the nonparametric bootstrap.
For $b=1,\dotsc,B$, let $(M^b_{1,n},\dotsc,M^b_{n,n})^{\top}$ be a multinomial random vector with success probabilities $(1/n,\dotsc,1/n)$.
Let $\widetilde{C}^b_{Y_1Y_2}$ and $\widetilde{C}^b_{Y_1^*Y_2^*}$ denote the bootstrap estimators of actual and counterfactual copulas, respectively.
To be specific,
\begin{align*}
\widetilde{C}^b_{Y_1Y_2}(u_1,u_2)&=\frac{1}{n}\sum_{i=1}^nM^b_{i,n}\1\{\widehat{F}^b_{Y_1}(Y_{1i})\leq u_1,\widehat{F}^b_{Y_2}(Y_{2i})\leq u_2\},\\
\widetilde{C}^b_{Y_1^*Y_2^*}(u_1,u_2)&=\frac{1}{n}\sum_{i=1}^nM^b_{i,n}W_{i,n}\1\{\widehat{F}^b_{Y_1^*}(Y_{1i})\leq u_1,\widehat{F}^b_{Y_2^*}(Y_{2i})\leq u_2\},
\end{align*}
where $\widehat{F}^b_{Y_j}(y)=n^{-1}\sum_{i=1}^nM^b_{i,n}\1\{Y_{ji}\leq y\}$ and $\widehat{F}^b_{Y_j^*}(y)=n^{-1}\sum_{i=1}^nM^b_{i,n}W_{i,n}\1\{Y_{ji}\leq y\}$ for $j\in\{1,2\}$. Also let $\widetilde{\Delta}_\nu=\nu(\widetilde{C}_{Y_1^*Y_2^*})-\nu(\widetilde{C}_{Y_1Y_2})$ and
$\widetilde{\Delta}^{b}_\nu=\nu(\widetilde{C}^b_{Y_1^*Y_2^*})-\nu(\widetilde{C}^b_{Y_1Y_2})$ be the original and bootstrap estimators of policy effect $\Delta_\nu$, respectively.
We construct the $(1-\alpha)$ confidence interval for the policy effect $\Delta_\nu$ as
\[
\left[\widetilde{\Delta}_\nu-\frac{Q_\nu^{1-\alpha,B}}{\sqrt{n}},\quad\widetilde{\Delta}_\nu+\frac{Q_\nu^{1-\alpha,B}}{\sqrt{n}}\right],
\]
where $Q_\nu^{1-\alpha,B}$ is the $(1-\alpha)$-th quantile of
\[
\left\{\left|\sqrt{n}(\widetilde{\Delta}^b_\nu-\widetilde{\Delta}_\nu)-\frac{\sqrt{n}}{B}\sum_{b'=1}^B(\widetilde{\Delta}^{b'}_\nu-\widetilde{\Delta}_\nu)\right|\right\}_{b=1}^B.
\]
We generate $B=1,000$ bootstrap replicates and set the nominal coverage level to be 95\%.
Note that following \cite{Kojadinovic2019}, we consider the centered version of $\sqrt{n}(\widetilde{\Delta}^b_\nu-\widetilde{\Delta}_\nu)$. The rationale behind centering is that the bootstrap processes, whether for copulas or their corresponding functionals, converge weakly to centered Gaussian processes as stated in \cref{THM2}. Hence, using centered replicates would always lead to better approximation in finite samples.
It is also worth noting that rather than applying large-sample approximation to these measures, we can calculate them analytically.
As shown in \cite{Meyer2013}, the bivariate Gaussianity of $C_{Y_1Y_2}$ and $C_{Y^*_1Y^*_2}$ implies that the association measures, mentioned in \cref{Model}, can be expressed in terms of the Pearson correlation coefficient $r$, which is the off-diagonal element of the correlation matrix.
Concretely, we have
\begin{align*}
\nu_\tau(r)&=\frac{2}{\pi}\arcsin(r),\\
\nu_\rho(r)&=\frac{6}{\pi}\arcsin\left(\frac{r}{2}\right),\\
\nu_\gamma(r)&=\frac{2}{\pi}\left(\arcsin\left(\frac{1+r}{2}\right)-\arcsin\left(\frac{1-r}{2}\right)\right),\\
\nu_\beta(r)&=\frac{2}{\pi}\arcsin(r).
\end{align*}

\begin{table}[t]
\begin{threeparttable}
\centering
\caption{Simulation results for estimators of association measures and policy effects.}
\label{table:measure}
\begin{tabular*}{\textwidth}{@{\extracolsep{\fill}}lrrrrrrrrr}
\toprule
		&	\multicolumn{3}{c}{MAE}	&	\multicolumn{3}{c}{RMSE}	&	\multicolumn{3}{c}{CR}	\\
\cmidrule(lr){2-4}	\cmidrule(lr){5-7}	\cmidrule(lr){8-10}
\hfill$n$	&	\multicolumn{1}{c}{100}	&	\multicolumn{1}{c}{200}	&	\multicolumn{1}{c}{400}	&	\multicolumn{1}{c}{100}	&	\multicolumn{1}{c}{200}	&	\multicolumn{1}{c}{400}	&	\multicolumn{1}{c}{100}	&	\multicolumn{1}{c}{200}	&	\multicolumn{1}{c}{400}	\\
\midrule
\multicolumn{10}{l}{\bf Kendall's tau}	\\
Actual	&	0.044	&	0.028	&	0.022	&	0.055	&	0.035	&	0.027	&	0.941	&	0.958	&	0.946	\\
Counterfactual	&	0.113	&	0.075	&	0.052	&	0.132	&	0.088	&	0.062	&	0.754	&	0.862	&	0.901	\\
Policy effect	&	0.110	&	0.074	&	0.049	&	0.116	&	0.079	&	0.053	&	0.779	&	0.863	&	0.938	\\
\multicolumn{10}{l}{\bf Spearman's rho}	\\
Actual	&	0.056	&	0.035	&	0.028	&	0.070	&	0.045	&	0.034	&	0.936	&	0.955	&	0.949	\\
Counterfactual	&	0.112	&	0.082	&	0.062	&	0.138	&	0.100	&	0.076	&	0.924	&	0.969	&	0.983	\\
Policy effect	&	0.092	&	0.072	&	0.054	&	0.105	&	0.082	&	0.061	&	0.939	&	0.985	&	0.988	\\
\multicolumn{10}{l}{\bf Gini's gamma}	\\
Actual	&	0.052	&	0.033	&	0.025	&	0.065	&	0.041	&	0.031	&	0.932	&	0.955	&	0.950	\\
Counterfactual	&	0.096	&	0.069	&	0.052	&	0.119	&	0.084	&	0.063	&	0.919	&	0.961	&	0.974	\\
Policy effect	&	0.077	&	0.059	&	0.043	&	0.086	&	0.066	&	0.048	&	0.916	&	0.964	&	0.979	\\
\multicolumn{10}{l}{\bf Blomqvist's beta}	\\
Actual	&	0.074	&	0.051	&	0.036	&	0.092	&	0.062	&	0.045	&	0.968	&	0.970	&	0.967	\\
Counterfactual	&	0.115	&	0.082	&	0.058	&	0.142	&	0.101	&	0.072	&	0.941	&	0.949	&	0.947	\\
Policy effect	&	0.084	&	0.062	&	0.043	&	0.097	&	0.070	&	0.050	&	0.969	&	0.955	&	0.933	\\
\bottomrule
\end{tabular*}
\begin{tablenotes}
\item \scriptsize{Note:
The results are obtained from 1,000 simulations and 1,000 bootstrap replicates in each simulation. The association measures and policy effects are calculated based on copula estimates on the grids $\{(i/100, j/100): i, j= 0,1,\dots,100\}$. The nominal coverage level is 95\%. See texts for more estimation details.}
\end{tablenotes}
\end{threeparttable}
\end{table}

\cref{table:measure} reports the mean absolute errors (MAE), root mean squared errors (RMSE), and confidence interval coverage rates (CR) for the estimators of actual and counterfactual association measures and the corresponding policy effects.
As expected from our theoretical analysis, both MAE and RMSE nearly halve as the sample size quadruples; furthermore, the empirical coverage rates are generally close to the pre-specified nominal level.
These results suggest that the bootstrap method performs well in small samples.

\section{Empirical Study}
\label{EmpiricalStudy}
We apply the counterfactual copula method to studying the role of college education in intergenerational income mobility in the United States.
Over the past decades, researchers have observed that the association between incomes across generations among college-educated children is weaker than that among less educated groups.
That there is a ``meritocratic power'' of college education reducing the influence of family origin has been supported by a number of empirical studies, for example \cite{Hout1988}, \cite{Torche2011}, and \cite{Chetty2020}.
While most of the early literature hinges on the intergenerational income elasticity,
this measure has its limitations as discussed in \citet{Mitnik2020}.
Alternatively, \cite{Chetty2014} advocate assessing positional mobility by intergenerational rank correlation, which is none other than Spearman's rho between parent and child income.
Besides, it is well recognized that the observed high degree of mobility among children with college education reflects a combination of selection bias (the types of students admitted) and causal effects (the returns to college education).
To address this confounding issue, \cite{Zhou2019} adopts a reweighting approach whereas \cite{Gregg2019} employ unconditional quantile regression with the inclusion of confounding factors. 
In view of these advances, we use the proposed copula method to investigate the heterogeneous effects of education on intergenerational income mobility.
This approach allows us to simultaneously control for confounding factors and explore various aspects of mobility, including Spearman's rho and other measures of association discussed in Section~\ref{Model}.

\subsection{Data}
\label{data}
The data we use are drawn from the Panel Study of Income Dynamics (PSID), one of the primary databases used in the literature on intergenerational mobility in the United States.\footnote{See \cite{Mazumder2018} for a comprehensive survey of the PSID.}
Beginning in 1968, the PSID is a longitudinal study of a representative sample of US individuals and their family units. The sample has been reinterviewed annually through 1997 and biennially thereafter to collect information on the individuals and their descendants.
As children from the original PSID families transition to adulthood and form their own households, the sample size grew from 4,800 family units in 1968 to almost double in 2019.
However, family-level variables collected in each wave of the study are contained in separate single-year family files, which need to be merged by a cross-year individual file containing all individual-level variables.
In this study, we use both family and individual files for the survey years 1968--2019. Additionally, we follow \cite{Lee2009} to exclude the Survey of Economic Opportunity (SEO) part of the PSID due to serious irregularities in the sampling of SEO respondents.

Like most of the literature, we construct permanent incomes of parents and children, $Y_1$ and $Y_2$, using at least 5-year averages of total family income across years.\footnote{See Table 2 of \cite{Mazumder2018} for a selected summary of related studies.} There are several advantages of using averages and using family income. First, averages of income likely provide a better measure of permanent income by reducing the impact of transitory shocks. Second, by using averaged data, we do not require individuals to report income every year. Third, using family income instead of individual one allows us to keep the unemployed in the analysis. Fourth, families with one high-income spouse will be properly treated as high-income families. To avoid heterogeneity in life cycle earnings profiles, we only average the total family income for parents and children aged 25 to 55 at the time of data collection. We also adjust the family income in all years to 2019 dollars using the CPI-U-RS prior to any calculations. Note that because of the use of family income for both parents and children, we cannot separate the income of children from that of their parents if they live together. Consequently, our sample consists of 3,895 households, with children being the head (or the spouse of the head) of another household in a subsequent reporting year.

In addition to income variables, we also consider several family characteristics $X$ that may affect income mobility. These covariates include gender, race, and years of schooling for the family head (parent), as well as gender, year of birth, and years of schooling for the child. \cref{table:data} presents the descriptive statistics for all variables under consideration.

\begin{table}[t]
\begin{threeparttable}
\centering
\caption{Descriptive statistics.}
\label{table:data}
\begin{tabular*}{\textwidth}{@{\extracolsep{\fill}}ld{4}d{4}d{4}d{4}d{4}}
\toprule
	&	\multicolumn{1}{c}{Mean}	&	\multicolumn{1}{c}{SD}	&	\multicolumn{1}{c}{Min}	&	\multicolumn{1}{c}{Median}	&	\multicolumn{1}{c}{Max}	\\
\midrule
\multicolumn{6}{l}{\bf Parent}	\\
Income (1000s)	&	83.518	&	54.739	&	6.292	&	74.006	&	852.418	\\
Male				&	0.891	&	0.312	&	0	&	1	&	1	\\
White			&	0.889	&	0.315	&	0	&	1	&	1	\\
Years of schooling	&	13.135	&	2.848	&	0	&	12	&	17	\\
\multicolumn{6}{l}{\bf Child}	\\
Income (1000s)	&	85.825	&	64.382	&	5.081	&	73.707	&	1134.731	\\
Male				&	0.486	&	0.500	&	0	&	0	&	1	\\
Year of birth	&	1968.122	&	10.500	&	1938	&	1968	&	1986	\\
Years of schooling	&	14.387	&	2.083	&	7	&	14	&	17	\\		
\bottomrule
\end{tabular*}
\begin{tablenotes}
\item \scriptsize{Note:
The sample consists of 3,895 PSID households. Parental characteristics refer to the family head.}
\end{tablenotes}
\end{threeparttable}
\end{table}

\subsection{Scenarios}
\label{scenarios}

We evaluate the effects of education on intergenerational income mobility in two counterfactual scenarios.
In the first scenario, children of low educational attainment are given more years of schooling such that all children receive at least a certain amount of education.
Specifically, the counterfactual covariates $X^*=(X^*_\text{cedu}, X^{*\top}_\text{other})^\top$ consist of the counterfactual children's years of schooling $X^*_\text{cedu}=\max\{X_\text{cedu},s\}$ with $X_\text{cedu}$ the actual value and $s$ the years of ``compulsory'' schooling varying from one experiment to another, and $X^*_\text{other}=X_\text{other}$ representing the other family characteristics, which are held constant.
Incrementally increasing the parameter $s$ from 13 to 16 (corresponding to the case where a college degree is conferred on all children), we can examine the heterogeneous effects of college education on income mobility by comparing the counterfactual measures of association with the actual ones.

In the second scenario, children with less educated parents are required to complete a four-year college degree.
To be specific, let $X^*_\text{cedu}=\max\{X_\text{cedu},16\}\cdot\operatorname{1}\{X_\text{pedu}\leq s'\}+X_\text{cedu}\cdot\operatorname{1}\{X_\text{pedu}>s'\}$, where $X_\text{pedu}$ is the parental years of schooling and $s'$ is the threshold.
By shifting $s'$ from 6 to 17,\footnote{Note that only 1.34\% (52 out of 3,895) of parents in the sample have less than 6 years of education.}
we can investigate the trade-off between the effectiveness of a policy and the number of children affected by this policy.
This trade-off may further provide policymakers with guidance on the public investments in education from the perspective on intergenerational income mobility.

\subsection{Results}
\label{result}

The results of our first counterfactual scenario are summarized in \cref{fig:yos}, where the actual, counterfactual, and subgroup measures of association across children's compulsory schooling parameter $s$ are plotted.\footnote{
By a subgroup measure of association, we mean the association measure for a subgroup of families with children's actual years of schooling above or equal to $s$.
}
As can be seen in Panel (a), the actual Spearman's rho, depicted by the horizontal dashed line, falls well within the 95\% confidence interval for the counterfactual Spearman's rho except for the case of $s=16$, in which all children complete a four-year college degree.
This finding indicates that, other things being equal, providing some college education to all children is unlikely to boost intergenerational income mobility,
while the receipt of a college degree can effectively reduce income persistence by about one-third of the magnitude.
In addition, the subgroup Spearman's rho for children who attend college without receiving a college degree substantially differs from the corresponding (i.e., $s=13, 14, 15$) counterfactual Spearman's rho.
Such differences suggest that the high mobility (low association) observed in each of these subgroups may be attributed to selection bias, rather than college education per se. 
These findings are robust to the use of the other measures of association, inclusive of Kendall's tau, Gini's gamma, and Blomqvist's beta, shown in Panels (b)-(d) respectively.

\begin{figure}
\centering
\begin{subfigure}{0.49\textwidth}
\centering
\includegraphics[width=\textwidth]{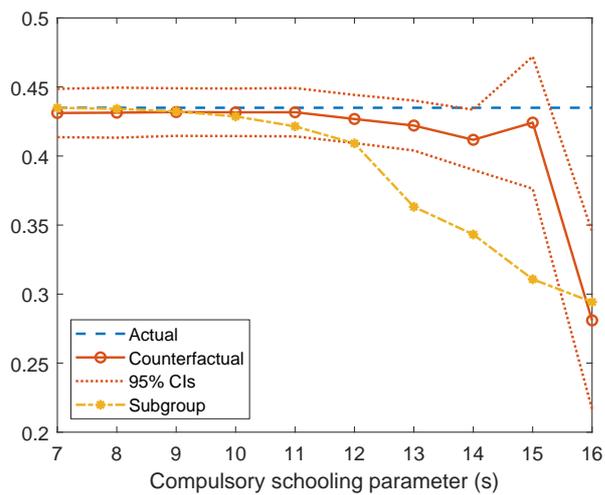}
\caption{Spearman's rho}
\end{subfigure}
\hfill
\begin{subfigure}{0.49\textwidth}
\centering
\includegraphics[width=\textwidth]{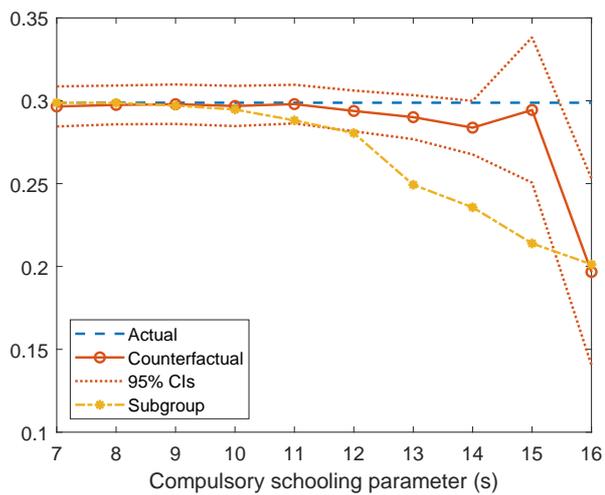}
\caption{Kendall's tau}
\end{subfigure}
\par\medskip
\begin{subfigure}{0.49\textwidth}
\centering
\includegraphics[width=\textwidth]{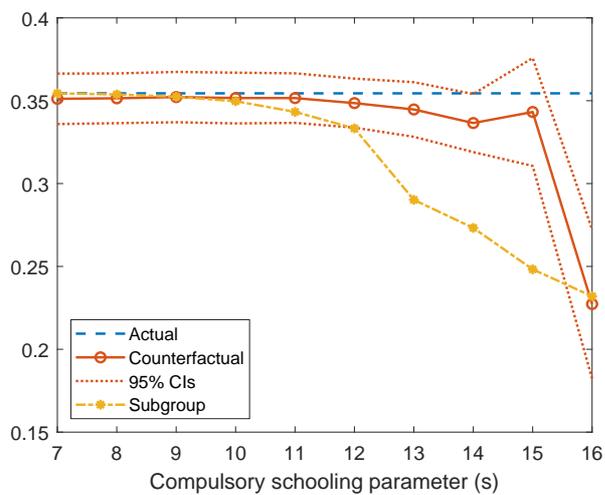}
\caption{Gini's gamma}
\end{subfigure}
\hfill
\begin{subfigure}{0.49\textwidth}
\centering
\includegraphics[width=\textwidth]{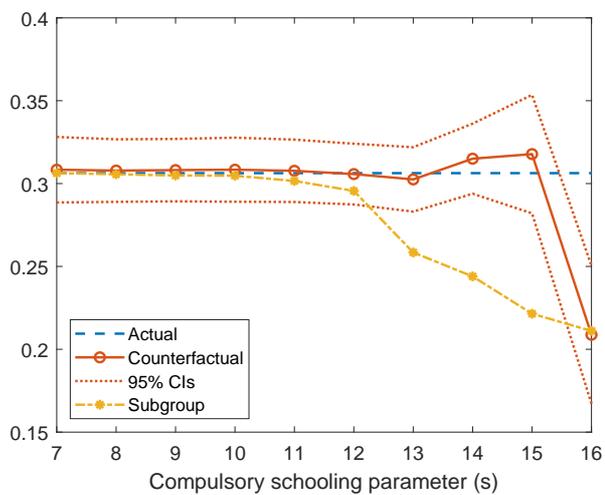}
\caption{Blomqvist's beta}
\end{subfigure}
\caption{Measures of association across children's compulsory schooling parameter.\label{fig:yos}}
\end{figure}

\cref{fig:quantile} shows the counterfactual policy effects on the four association measures in the second scenario.
As shown in Panel (a), the policy effect on Spearman's rho is significantly less than zero when the threshold parameter is set at $s'=10$.
In this setting, children with parental eduction less than, or equal to, ten years of schooling are required to complete a four-year college degree,
and those families affected by the counterfactual policy account for 14.87\% (579 out of 3,895) of the sample.
In addition, when the threshold parameter $s'$ is greater than 12, the relatively flat curve is suggestive of a negligible improvement in the policy effect on Spearman's rho.
The robustness of these findings are confirmed by the similar curves plotted in Panels (b)-(d) where the other measures of association are adopted.
Taken together, these results demonstrate that intergenerational income mobility can be effectively improved if children from less educated families (i.e., parental eduction less than or equal to ten years) are required to complete a four-year college degree.

\begin{figure}
\begin{subfigure}{0.49\textwidth}
\centering
\includegraphics[width=\textwidth]{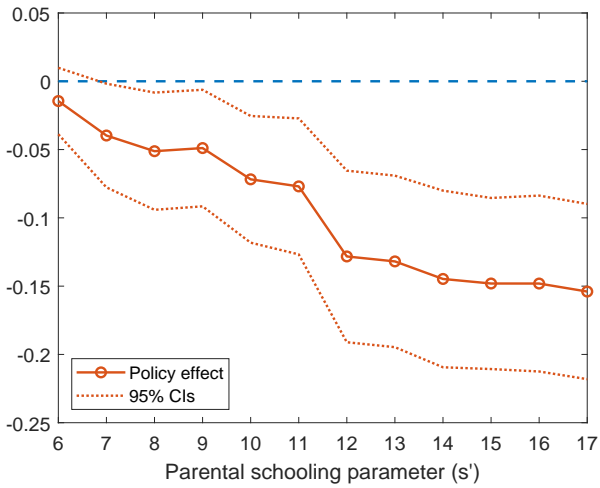}
\caption{Spearman's rho}
\end{subfigure}
\hfill
\begin{subfigure}{0.49\textwidth}
\centering
\includegraphics[width=\textwidth]{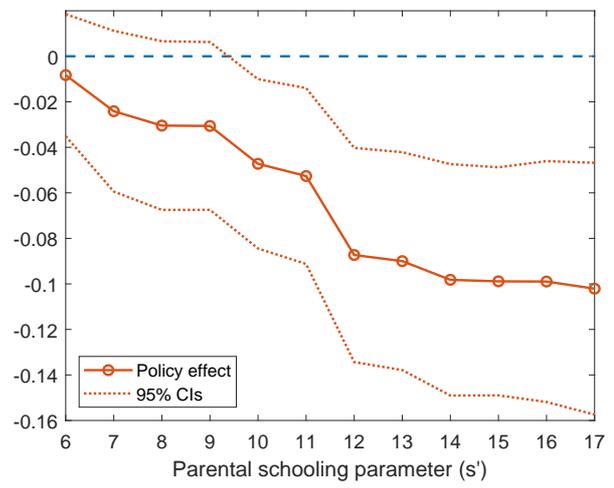}
\caption{Kendall's tau}
\end{subfigure}
\par\medskip
\begin{subfigure}{0.49\textwidth}
\centering
\includegraphics[width=\textwidth]{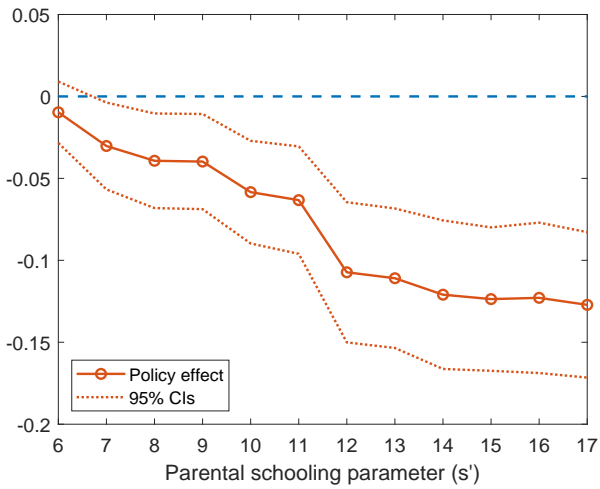}
\caption{Gini's gamma}
\end{subfigure}
\hfill
\begin{subfigure}{0.49\textwidth}
\centering
\includegraphics[width=\textwidth]{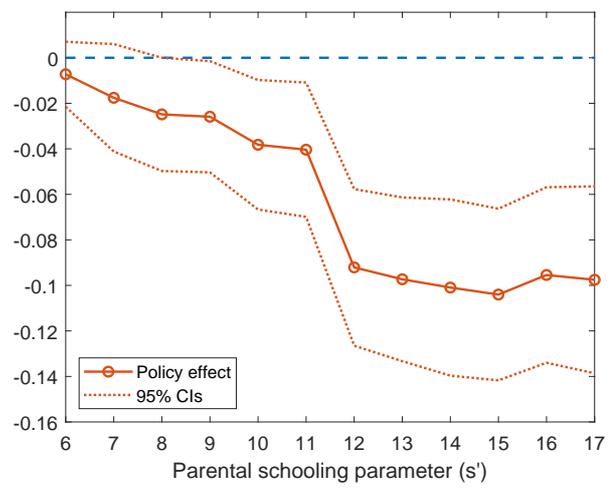}
\caption{Blomqvist's beta}
\end{subfigure}
\caption{Policy effects of a college degree for children with less educated parents on measures of association.\label{fig:quantile}}
\end{figure}

\section{Conclusion}
\label{Conclusion}
In this paper, we propose a nonparametric estimator of the counterfactual copula in response to exogenous policy intervention and prove its weak convergence.
We also establish bootstrap validity for inference on the counterfactual copula and the corresponding measures of association. Simulation results are in line with our theoretical findings.
Finally, we illustrate the use of the proposed method with an empirical application investigating the role of college education in intergenerational income mobility.
Our empirical results suggest that if children from less-educated families were given a college
degree, then intergenerational income mobility would be substantially improved.

\appendix


\begin{appendices}

\numberwithin{equation}{section}

\section{Technical proofs}
\label{Appendix}
\subsection{Proof of \cref{PRO1}}
The joint CDF of $(Y_1^*,Y_2^*)^\top$ can be identified because we have
\begin{align*}
F_{Y_1^*Y_2^*}(y_1,y_2)
&= \E[\1\{Y_1^*\leq y_1,Y_2^*\leq y_2\}]\\
&= \int \E[\1\{g_1(x,\varepsilon_1)\leq y_1,g_2(x,\varepsilon_2)\leq y_2\}|X^*=x] F_{X^*}(\d x)\\
&= \int \E[\1\{g_1(x,\varepsilon_1)\leq y_1,g_2(x,\varepsilon_2)\leq y_2\}|X=x] F_{X^*}(\d x)\\
&= \int F_{Y_1Y_2|X}(y_1,y_2|x)F_{X^*}(\d x),
\end{align*}
where the third equality holds by Assumptions~\ref{I1} and~\ref{I2}. Following similar arguments, we can also write $F_{Y_1^*}(y_1)$ and $F_{Y_2^*}(y_2)$ as the stated formulas. Since $F_{Y_1^*}(y_1)$ and $F_{Y_2^*}(y_2)$ are continuous by Assumption~\ref{I3}, the desired result follows from Sklar's theorem.

\subsection{Proof of \cref{PRO3}}
We decompose the counterfactual copula process as follows. For each $u_1$ and $u_2$ in $(0,1)$,
\begin{align*}
\widehat{\mathbb{C}}_{Y^*_1Y^*_2}(u_1,u_2)
&=\sqrt{n}(F_{Y^*_1Y^*_2}(\widehat{F}^{-1}_{Y^*_1}(u_1),\widehat{F}^{-1}_{Y^*_2}(u_2))-F_{Y^*_1Y^*_2}(F^{-1}_{Y^*_1}(u_1),F^{-1}_{Y^*_2}(u_2)))\\
&\quad+\sqrt{n}(\widehat{F}_{Y^*_1Y^*_2}(\widehat{F}^{-1}_{Y^*_1}(u_1),\widehat{F}^{-1}_{Y^*_2}(u_2))-F_{Y^*_1Y^*_2}(\widehat{F}^{-1}_{Y^*_1}(u_1),\widehat{F}^{-1}_{Y^*_2}(u_2)))\\
&\equiv\text{Term 1}+\text{Term 2}.
\end{align*}
First, we focus on Term 1. Theorem 20.8 and Lemma 21.4 of \cite{vanderVaart1998} imply that for each $j\in\{1,2\}$ and $u\in(0,1)$,
\begin{equation}
\label{AA1}
\sqrt{n}(\widehat{F}_{Y_j^*}^{-1}(u)-F_{Y_j^*}^{-1}(u))
=\frac{-1}{f_{Y_j^*}(F^{-1}_{Y_j^*}(u))}\sqrt{n}(\widehat{F}_{Y_j^*}(F^{-1}_{Y_j^*}(u))-u)+r_{j,n}(u),
\end{equation}
where the remainder satisfies $\sup_{u\in(0,1)}|r_{j,n}(u)|=o_p(1)$.
The delta method implies that
\begin{align}
\label{AA2}
&\sqrt{n}(F_{Y^*_1Y^*_2}(\widehat{F}^{-1}_{Y^*_1}(u_1),\widehat{F}^{-1}_{Y^*_2}(u_2))-F_{Y^*_1Y^*_2}(F^{-1}_{Y^*_1}(u_1),F^{-1}_{Y^*_2}(u_2)))\notag\\
={}&\sum_{j=1}^2\frac{\partial F_{Y_1^*Y_2^*}(F_{Y_1^*}^{-1}(u_1),F_{Y_2^*}^{-1}(u_2))}{\partial y_j}\sqrt{n}(\widehat{F}_{Y_j^*}^{-1}(u_j)-F_{Y_j^*}^{-1}(u_j))+\gamma_n(u_1,u_2)
\end{align}
for some remainder $\gamma_n$. It follows from \cref{LEM2} that there is a constant $c_0$ such that
\begin{align}
\label{AA3}
\sup_{(u_1,u_2)\in(0,1)^2}|\gamma_n(u_1,u_2)|
&\leq c_0\sqrt{n}\left(\sup_{u\in(0,1)}|\widehat{F}_{Y_1^*}^{-1}(u)-F_{Y_1^*}^{-1}(u)|^2
+\sup_{u\in(0,1)}|\widehat{F}_{Y_2^*}^{-1}(u)-F_{Y_2^*}^{-1}(u)|^2\right)\notag\\
&=O_p(n^{-1/2}).
\end{align}
Combining \cref{AA1,AA2,AA3}, we obtain
\begin{align}
\label{AA4}
&\sqrt{n}(F_{Y^*_1Y^*_2}(\widehat{F}^{-1}_{Y^*_1}(u_1),\widehat{F}^{-1}_{Y^*_2}(u_2))-F_{Y^*_1Y^*_2}(F^{-1}_{Y^*_1}(u_1),F^{-1}_{Y^*_2}(u_2)))\notag\\
={}&-\sum_{j=1}^2\frac{1}{f_{Y_j^*}(F^{-1}_{Y_j^*}(u_j))}
\frac{\partial F_{Y_1^*Y_2^*}(F_{Y_1^*}^{-1}(u_1),F_{Y_2^*}^{-1}(u_2))}{\partial y_j}\sqrt{n}(\widehat{F}_{Y_j^*}(F^{-1}_{Y_j^*}(u_j))-u_j)\notag\\
&+\sum_{j=1}^2\frac{\partial F_{Y_1^*Y_2^*}(F_{Y_1^*}^{-1}(u_1),F_{Y_2^*}^{-1}(u_2))}{\partial y_j}r_{j,n}(u_j)+\gamma_n(u_1,u_2)\notag\\
={}&-\sum_{j=1}^2\frac{\partial C_{Y^*_1Y^*_2}(u_1,u_2)}{\partial u_j}
\sqrt{n}(\widetilde{G}_{Y^*_j}(u_j)-u_j)+R^*_{1,n}(u_1,u_2),
\end{align}
where the remainder term
\[
R^*_{1,n}(u_1,u_2)\equiv \sum_{j=1}^2
\frac{\partial F_{Y_1^*Y_2^*}(F_{Y_1^*}^{-1}(u_1),F_{Y_2^*}^{-1}(u_2))}{\partial y_j}
r_{j,n}(u_j)+\gamma_n(u_1,u_2)
\]
satisfies $\sup_{(u_1,u_2)\in(0,1)^2}|R^*_{1,n}(u_1,u_2)|=o_p(1)$.

Next, we focus on Term 2. Let $\mathcal{D}_n\equiv\{(Y_{1i},Y_{2i},X_i)\}_{i=1}^n$. The weak convergence of $\sqrt{n}(\widehat{F}_{Y_1^*Y_2^*}-F_{Y_1^*Y_2^*})$ established in \cref{LEM1} implies that the class
\[
\{\mathcal{D}_n\mapsto\sqrt{n}(\widehat{F}_{Y_1^*Y_2^*}(y_1,y_2)
-F_{Y_1^*Y_2^*}(y_1,y_2)):(y_1,y_2)\in \Y_1\Y_2\}
\]
is stochastic equicontinuous.
In addition, $\sup_{u\in(0,1)}|\widehat{F}_{Y_j^*}^{-1}(u)-F_{Y_j^*}^{-1}(u)|=o_p(1)$ for each $j\in\{1,2\}$ by \cref{LEM2}. It follows that
\begin{align}
\label{AA5}
&\sqrt{n}(\widehat{F}_{Y_1^*Y_2^*}(\widehat{F}_{Y_1^*}^{-1}
(u_1),\widehat{F}_{Y_2^*}^{-1}(u_2))-F_{Y_1^*Y_2^*}(\widehat{F}_{Y_1^*}^{-1}(u_1),\widehat{F}_{Y_2^*}^{-1}(u_2)))\notag\\
={}&\sqrt{n}(\widehat{F}_{Y_1^*Y_2^*}(F_{Y_1^*}^{-1}(u_1),F_{Y_2^*}^{-1}(u_2))-F_{Y_1^*Y_2^*}(F_{Y_1^*}^{-1}(u_1),F_{Y_2^*}^{-1}(u_2)))+R^*_{2,n}(u_1,u_2)\notag\\
={}&\sqrt{n}(\widetilde{G}_{Y^*_1Y^*_2}(u_1,u_2)-C_{Y^*_1Y^*_2}(u_1,u_2))+R^*_{2,n}(u_1,u_2),
\end{align}
where the remainder term satisfies $\sup_{u_1,u_2}|R^*_{2,n}(u_1,u_2)|=o_p(1)$. Combining \cref{AA4,AA5} completes the proof.

\subsection{Proof of \cref{THM1}}
First, we prove the weak convergence of $(\widehat{\mathbb{C}}_{Y_1Y_2},\widehat{\mathbb{C}}_{Y^*_1Y^*_2})^\top$ by showing that of $\widehat{\mathbb{C}}_{Y_1Y_2}$ and $\widehat{\mathbb{C}}_{Y^*_1Y^*_2}$, respectively. On the one hand, the decomposition of $\widehat{\mathbb{C}}_{Y_1Y_2}$ in \cref{PRO2} implies that the weak convergence of $\widehat{\mathbb{C}}_{Y_1Y_2}$ is guaranteed by that of $\sqrt{n}(\widetilde{G}_{Y_1Y_2}-C_{Y_1Y_2})$, which is valid by Example 2.1.3 of \cite{vanderVaartWellner1996}.
On the other hand, we aim to establish the weak convergence of $\widehat{\mathbb{C}}_{Y^*_1Y^*_2}$. By \cref{PRO3} and its subsequent discussion, it is sufficient to show the weak convergence of $\sqrt{n}(\widehat{F}_{Y_1^*Y_2^*}-F_{Y_1^*Y_2^*})$ in $(\ell^\infty(\Y_1\Y_2),\|\cdot\|_\infty)$. This latter weak convergence follows from the decomposition of $\sqrt{n}(\widehat{F}_{Y_1^*Y_2^*}-F_{Y_1^*Y_2^*})$ in \cref{LEM1}, as in \cite{Rothe2010}.

Next, applying Theorem 20.8 of \cite{vanderVaart1998}, we have
\[
(\sqrt{n}(\nu(\widehat{C}_{Y_1Y_2})-\nu(C_{Y_1Y_2})),\sqrt{n}(\nu(\widehat{C}_{Y^*_1Y^*_2})-\nu(C_{Y^*_1Y^*_2})))^\top\Rightarrow \nu'_{(C_{Y_1Y_2},C_{Y^*_1Y^*_2})}(\mathbb{C})\equiv\mathbb{V},
\]
where $\nu'_{(C_{Y_1Y_2},C_{Y^*_1Y^*_2})}$ is the Hadamard derivative of $\nu$ at $(C_{Y_1Y_2},C_{Y^*_1Y^*_2})$. Since the Hadamard derivative $\nu'_{(C_{Y_1Y_2},C_{Y^*_1Y^*_2})}$ is continuous and linear, the centered Gaussianity of $\mathbb{C}$ implies that of $\mathbb{V}$.

\subsection{Proof of \cref{THM2}}
\cref{THM1} shows that the counterfactual copula process $\widehat{\mathbb{C}}_{Y^*_1Y^*_2}$ is asymptotically equivalent to an empirical process of the form
\[
f\in\mathcal{F}^*\mapsto\frac{1}{\sqrt{n}}\sum_{i=1}^n(f(Z_i)-\E[f(Z_i)]),
\]
where $Z_i\equiv(Y_{1i},Y_{2i},X_i,X^*_i)$ and $\mathcal{F}^*$ is a class of functions defined on $[0,1]^2$. By \cref{PRO3,LEM1}, $\mathcal{F}^*$ is a Donsker class. Similarly, it can be shown that the empirical copula process $\widehat{\mathbb{C}}_{Y_1Y_2}$
is asymptotically equivalent to an empirical process of the form
\[
f\in\mathcal{F}\mapsto\frac{1}{\sqrt{n}}\sum_{i=1}^n(f(Z_i)-\E[f(Z_i)]),
\]
where $\mathcal{F}$ is a Donsker class of functions defined on $[0,1]^2$. The results follow from Theorems 3.6.1 and 3.9.11 of \cite{vanderVaartWellner1996}.

\subsection{Lemmas}

\begin{lem}
\label{LEM1}
If Assumptions {\normalfont I}, {\normalfont K1-K3}, {\normalfont B}, and {\normalfont S1-S5} hold, then for each $(y_1,y_2)\in\Y_1\Y_2$,
\begin{align*}
&\sqrt{n}(\widehat{F}_{Y_1^*Y_2^*}(y_1,y_2)-F_{Y_1^*Y_2^*}(y_1,y_2))\\
={}&\frac{1}{\sqrt{n}}\sum_{i=1}^n(F_{Y_1Y_2|X}(y_1,y_2|X^*_i)-F_{Y_1^*Y_2^*}(y_1,y_2))\\
&+\frac{1}{\sqrt{n}}\sum_{i=1}^n\frac{f_{X^*}(X_i)}{f_{X}(X_i)}(
\1\{Y_{1i}\leq y_1,Y_{2i}\leq y_2\}-F_{Y_1Y_2|X}(y_1,y_2|X_i))+r^*_n(y_1,y_2),
\end{align*}
where the remainder term $r^*_n$ satisfies
\begin{align*}
\sup_{(y_1,y_2)\in\Y_1\Y_2}|r^*_n(y_1,y_2)|=o_p(1).
\end{align*}
\end{lem}

\begin{proof}
Let $\mathcal{D}_n\equiv\{(Y_{1i},Y_{2i},X_i)\}_{i=1}^n$ and $X^*$ be a random variable that is independent of $\mathcal{D}_n$ and identically distributed as $X^*_1$. Simple algebra shows that for each $(y_1,y_2)\in\Y_1\Y_2$,
\begin{align*}
&\sqrt{n}(\widehat{F}_{Y_1^*Y_2^*}(y_1,y_2)-F_{Y_1^*Y_2^*}(y_1,y_2))\\
={}&\frac{1}{\sqrt{n}}\sum_{i=1}^n
(\widehat{F}_{Y_1Y_2|X}(y_1,y_2|X_i^*)-F_{Y_1^*Y_2^*}(y_1,y_2))\\
={}&\frac{1}{\sqrt{n}}\sum_{i=1}^n
(F_{Y_1Y_2|X}(y_1,y_2|X_i^*)-F_{Y_1^*Y_2^*}(y_1,y_2))\\
&+\sqrt{n}(\E[\widehat{F}_{Y_1Y_2|X}(y_1,y_2|X^*)|\mathcal{D}_n]
-\E[F_{Y_1Y_2|X}(y_1,y_2|X^*)])\\
&+\frac{1}{\sqrt{n}}\sum_{i=1}^n((\widehat{F}_{Y_1Y_2|X}(y_1,y_2|X_i^*)-\E[\widehat{F}_{Y_1Y_2|X}(y_1,y_2|X_i^*)|\mathcal{D}_n])\\
&\hspace{2cm}-(F_{Y_1Y_2|X}(y_1,y_2|X_i^*)-\E[F_{Y_1Y_2|X}(y_1,y_2|X_i^*)]))\\
\equiv{}&\text{Term I}+\text{Term II}+\text{Term III}.
\end{align*}
Our proof builds on the strategy of Theorem 1 in \cite{Rothe2010} by discussing the second and third term on the right-hand side above.
We start by showing that Term III is asymptotically negligible. Let $\widehat{\Gamma}(y_1,y_2|x;h_n)\equiv\widehat{F}_{Y_1Y_2|X}(y_1,y_2|x;h_n)-F_{Y_1Y_2|X}(y_1,y_2|x)$ for $(y_1,y_2,x)\in\Y_1\Y_2\X$.
Theorem 8 of \cite{Hansen2008} implies that for each $(y_1,y_2)\in\Y_1\Y_2$,
\[
\E\left[\widehat{\Gamma}(y_1,y_2|X^*;h_n)^2|\mathcal{D}_n\right]\leq \sup_{x\in\X}\left|\widehat{\Gamma}(y_1,y_2|x;h_n)\right|^2=O_p\left(\frac{\log{n}}{nh_n^d}+h_n^{2r}\right)=o_p(1).
\]
The last equality holds by Assumptions~\ref{B1} and~\ref{B2}.
Let $\mathcal{C}_r(\X)$ be the class of real-valued functions defined on $\X$ whose partial derivatives up to order $r$ exist and are bounded by some constant.
Example 19.9 of \cite{vanderVaart1998} shows that $\mathcal{C}_r(\X)$ is Donsker whenever $r>d/2$, which is guaranteed under Assumptions~\ref{B2} and~\ref{B3}.
Note that $\{\tilde{x}\mapsto\widehat{\Gamma}(y_1,y_2
|\tilde{x};h_n):(y_1,y_2)\in\Y_1\Y_2\}\subseteq \mathcal{C}_r(\X)$ under Assumptions~\ref{K3} and~\ref{S2}.
Following the arguments in Lemma 19.24 of \cite{vanderVaart1998}, we obtain
\[
\sup_{(y_1,y_2)\in\Y_1\Y_2}
\left|\frac{1}{\sqrt{n}}\sum_{i=1}^n\widehat{\Gamma}(y_1,y_2|X_i^*;h_n)-\E[\widehat{\Gamma}(y_1,y_2|X^*_i;h_n)|\mathcal{D}_n]\right|=o_p(1).
\]

Next, we decompose Term II as follows.
\begin{align*}
&\sqrt{n}(\E[\widehat{F}_{Y_1Y_2|X}(y_1,y_2|X^*)|\mathcal{D}_n]
-\E[F_{Y_1Y_2|X}(y_1,y_2|X^*)])\\
={}&\sqrt{n}\int_{\X^*}(\widehat{F}_{Y_1Y_2|X}(y_1,y_2|x)
-F_{Y_1Y_2|X}(y_1,y_2|x))f_{X^*}(x)\d x\\
={}&\int_{\X^*}\left(\frac{1}{\sqrt{n}h_n^d}\sum_{i=1}^n
(\1\{Y_{1i}\leq y_1,Y_{2i}\leq y_2\}-F_{Y_1Y_2|X}(y_1,y_2|x))K\left(\frac{X_i-x}{h_n}\right)\right)
\frac{f_{X^*}(x)}{\hat{f}_{X}(x)}\d x\\
={}&\frac{1}{\sqrt{n}h_n^d}\sum_{i=1}^n(\1\{Y_{1i}\leq y_1,Y_{2i}\leq y_2\}-F_{Y_1Y_2|X}(y_1,y_2|X_i))\int_{\X^*}K\left(\frac{X_i-x}{h_n}\right)\frac{f_{X^*}(x)}{\hat{f}_{X}(x)}\d x\\
&+\frac{1}{\sqrt{n}h_n^d}\sum_{i=1}^n\int_{\X^*}
(F_{Y_1Y_2|X}(y_1,y_2|X_i)-F_{Y_1Y_2|X}(y_1,y_2|x))K\left(\frac{X_i-x}{h_n}\right)\frac{f_{X^*}(x)}{\hat{f}_{X}(x)}\d x\\
\equiv{}&\text{Term II.a}+\text{Term II.b},
\end{align*}
where $\hat{f}_{X}$ is the kernel density estimator; that is,
\begin{align*}
\hat{f}_{X}(x)\equiv\frac{1}{nh_{n}^d}\sum_{i=1}^{n}K\left(\frac{X_{i}-x}{h}\right).
\end{align*}
Under Assumptions~\ref{S1}, \ref{S3}, and \ref{K1}-\ref{K3}, the second order Taylor series expansion of $1/\hat{f}_{X}(x)$ around $1/f_{X}(x)$ implies that
\begin{align*}
&\text{Term II.a}\\
={}&\frac{1}{\sqrt{n}h_n^d}\sum_{i=1}^n(\1\{Y_{1i}\leq y_1,Y_{2i}\leq y_2\}-F_{Y_1Y_2|X}(y_1,y_2|X_i))\int_{\X^*}K\left(\frac{X_i-x}{h_n}\right)\frac{f_{X^*}(x)}{f_{X}(x)}\d x\\
&+\frac{1}{\sqrt{n}h_n^d}\sum_{i=1}^n(\1\{Y_{1i}\leq y_1,Y_{2i}\leq y_2\}-F_{Y_1Y_2|X}(y_1,y_2|X_i))\\
&\pushright{\cdot\int_{\X^*}K\left(\frac{X_i-x}{h_n}\right)\frac{f_{X^*}(x)}{f_{X}(x)^2}(f_{X}(x)-\hat{f}_{X}(x))\d x
+o_p(1)}\\
={}&\frac{1}{\sqrt{n}}\sum_{i=1}^n(\1\{Y_{1i}\leq y_1,Y_{2i}\leq y_2\}-F_{Y_1Y_2|X}(y_1,y_2|X_i))
\frac{f_{X^*}(X_i)}{f_{X}(X_i)}\\
&+\frac{1}{\sqrt{n}h_n^d}\sum_{i=1}^n\left(\1\{Y_{1i}\leq y_1,Y_{2i}\leq y_2\}-F_{Y_1Y_2|X}(y_1,y_2|X_i)\right)\\
&\pushright{\cdot\int_{\X^*}K\left(\frac{X_i-x}{h_n}\right)\frac{f_{X^*}(x)}{f_{X}(x)^2}(f_{X}(x)-\hat{f}_{X}(x))\d x+o_p(1).}
\end{align*}
We are now in a position to show that
\begin{align}
\triangle_n&\equiv\frac{1}{\sqrt{n}h_n^d}\sum_{i=1}^n(\1\{Y_{1i}\leq y_1,Y_{2i}\leq y_2\}-F_{Y_1Y_2|X}(y_1,y_2|X_i))\notag\\
&\pushright{\cdot\int_{\X^*}K\left(\frac{X_i-x}{h_n}\right)
\frac{f_{X^*}(x)}{f_{X}(x)^2}(f_{X}(x)-\hat{f}_{X}(x))\d x}\label{delta}
\end{align}
is asymptotically uniformly negligible. Note that for each $i\in\{1,\ldots,n\}$,
\begin{align}
&\frac{1}{h_n^d}\int_{\X^*}K\left(\frac{X_i-x}{h_n}\right)
\frac{f_{X^*}(x)}{f_{X}(x)^2}(f_{X}(x)-\hat{f}_{X}(x))\d x\notag \\
={}&\frac{f_{X^*}(X_i)}{f_{X}(X_i)}+O_p(h_n^r)-\frac{1}{nh_n^d}\sum_{j=1}^n\int K(w)\frac{f_{X^*}(X_i+h_nw)}{f_{X}(X_i+h_nw)^2}
K\left(\frac{X_i-X_j}{h_n}+w\right)\d w\notag\\
={}&\frac{f_{X^*}(X_i)}{f_{X}(X_i)}-\frac{1}{nh_n^d}\sum_{j=1}^n
\frac{f_{X^*}(X_i)}{f_{X}(X_i)^2}K\left(\frac{X_i-X_j}{h_n}\right)
-\frac{1}{n}\sum_{j=1}^nQ(X_i,X_j;h_n)+o_p(n^{-1/2}),\label{int}
\end{align}
where
\[
Q(x_1,x_2;h_n)\equiv\frac{1}{h_n^d}\int K(z)\left(
\frac{f_{X^*}(x_1+h_nz)}{f_{X}(x_1+h_nz)^2}
K\left(\frac{x_1-x_2}{h_n}+z\right)
-\frac{f_{X^*}(x_1)}{f_{X}(x_1)^2}
K\left(\frac{x_1-x_2}{h_n}\right)\right)\d z.
\]
For ease of exposition, we introduce the following notation.
Let $\bar{Q}(x;h_n)\equiv\E[Q(x,X;h_n)]$, $\bar{f}(x;h_n)\equiv\E\left[\frac{1}{h_n^d}K\left(\frac{x-X}{h_n}\right)\right]$.
In addition, for each $(y_1,y_2)\in\Y_1\Y_2$, let
\begin{align}
\label{Lfun}
&\mathcal{L}(W_i,W_j;y_1,y_2,h_n)\equiv(\1\{Y_{1i}\leq y_1,Y_{2i}\leq y_2\}-F_{Y_1Y_2|X}(y_1,y_2|X_i))\notag\\
&\pushright{\cdot\frac{f_{X^*}(X_i)}{f_{X}(X_i)^2}
\left(h_n^d\bar{f}(X_i;h_n)-K\left(\frac{X_i-X_j}{h_n}\right) \right)}
\end{align}
and
\begin{align}
\label{Mfun}
&\mathcal{M}(W_i,W_j;y_1,y_2,h_n)\equiv(\1\{Y_{1i}\leq y_1,Y_{2i}\leq y_2\}-F_{Y_1Y_2|X}(y_1,y_2|X_i))\notag\\
&\pushright{\cdot h_n^d(Q(X_i,X_j;h_n)-\bar{Q}(X_i;h_n))}
\end{align}
where $W_i\equiv(Y_{1i},Y_{2i},X_i)^\top$ for each $i\in\{1,2,\ldots,n\}$. Substituting (\ref{int}) into (\ref{delta}), we obtain
\begin{align*}
\triangle_n
&=\frac{1}{n^{3/2}h^d_n}\sum_{i=1}^n\sum_{j=1}^n
(\1\{Y_{1i}\leq y_1,Y_{2i}\leq y_2\}-F_{Y_1Y_2|X}(y_1,y_2|X_i))\\
&\pushright{\cdot\frac{f_{X^*}(X_i)}{f_{X}(X_i)^2}
\left(h^d_nf_{X}(X_i)-K\left(\frac{X_i-X_j}{h_n}\right)\right)}\\
&\quad-\frac{1}{n^{3/2}}\sum_{i=1}^n\sum_{j=1}^n
(\1\{Y_{1i}\leq y_1,Y_{2i}\leq y_2\}-F_{Y_1Y_2|X}(y_1,y_2|X_i))Q(X_i,X_j;h_n)+o_p(1)\\
&=\frac{1}{n^{3/2}h^d_n}\sum_{i=1}^n\sum_{j=1}^n\mathcal{L}(W_i,W_j;y_1,y_2,h_n)-\frac{1}{n^{1/2}}\sum_{i=1}^n\sum_{j=1}^n\mathcal{M}(W_i,W_j;y_1,y_2,h_n)\\
&\quad+\frac{1}{n^{1/2}}\sum_{i=1}^n(\1\{Y_{1i}\leq y_1,Y_{2i}\leq y_2\}-F_{Y_1Y_2|X}(y_1,y_2|X_i))\frac{f_{X^*}(X_i)}{f_{X}(X_i)^2}
(f_{X}(X_i)-\bar{f}(X_i;h_n))\\
&\quad-\frac{1}{n^{1/2}}\sum_{i=1}^n(\1\{Y_{1i}\leq y_1,Y_{2i}\leq y_2\}-F_{Y_1Y_2|X}(y_1,y_2|X_i))\bar{Q}(X_i;h_n)+o_p(1).
\end{align*}
Note that for each $(y_1,y_2)\in\Y_1\Y_2$ and $h>0$,
$\frac{1}{n(n-1)}\sum_{i=1}^n\sum_{j\neq i}\mathcal{L}(W_i,W_j;y_1,y_2,h)$ and
$\frac{1}{n(n-1)}\sum_{i=1}^n\sum_{j\neq i}\mathcal{M}(W_i,W_j;y_1,y_2,h)$
are both degenerate U statistics of order 2 because with probability one, $\E[\mathcal{L}(W_i,W_j;y_1,y_2,h)|W_i]=\E[\mathcal{L}(W_i,W_j;y_1,y_2,h)|W_j]=0$ and $\E[\mathcal{M}(W_i,W_j;y_1,y_2,h)|W_i]=\E[\mathcal{M}(W_i,W_j;y_1,y_2,h)|W_j]=0$. \cref{LEM3} shows that both $\{\mathcal{L}(\cdot,\cdot;y_1,y_2,h):(y_1,y_2)\in \Y_1\Y_2,h>0\}$ and $\{\mathcal{M}(\cdot,\cdot;y_1,y_2,h):(y_1,y_2)\in \Y_1\Y_2,h\in(0,1)\}$ are Euclidean classes. It follows from Corollary 4 of \cite{Sherman1994} that
\[
\sup_{(y_1,y_2,h)\in \Y_1\Y_2\times(0,\infty)}\left|\frac{1}{n}\sum_{i=1}^n\sum_{j\neq i}\mathcal{L}(W_i,W_j;y_1,y_2,h)\right|
=O_p(1)
\]
and
\[
\sup_{(y_1,y_2,h)\in \Y_1\Y_2\times(0,1)}\left|\frac{1}{n}\sum_{i=1}^n\sum_{j\neq i}\mathcal{M}(W_i,W_j;y_1,y_2,h)\right|=O_p(1).
\]
Finally, simple algebra shows that Term II.b is asymptotically uniformly negligible.
\end{proof}

\begin{lem}
\label{LEM2}
Suppose that the assumptions of \cref{LEM1} are valid. If, in addition, Assumptions~\ref{S6} and \ref{S7} hold, then for each $j\in\{1,2\}$,
\[
\sup_{u\in(0,1)}\left|\sqrt{n}(\widehat{F}_{Y_j^*}^{-1}(u)-F_{Y_j^*}^{-1}(u))\right|=O_p(1).
\]
\end{lem}

\begin{proof}
Fix $j\in\{1,2\}$. Equation \cref{AA1} in the proof of \cref{PRO3} implies that
\[
\sup_{u\in(0,1)}\left|\sqrt{n}(\widehat{F}_{Y_j^*}^{-1}(u)-F_{Y_j^*}^{-1}(u))\right|
\leq c_0\sup_{y\in\Y^*_j}\left|\sqrt{n}(\widehat{F}_{Y_j^*}\left(y\right)-F_{Y_j^*}(y))\right|+o_p(1)
\]
for some constant $c_0$ under Assumption~\ref{S7}.
It follows from \cref{LEM1} that
\begin{align*}
&\sup_{u\in(0,1)}\left|\sqrt{n}
\left[\widehat{F}_{Y_j^*}^{-1}(u)-F_{Y_j^*}^{-1}(u)\right]\right|\\
\leq{}& c_0\sup_{y\in\Y^*_j}
\left|\frac{1}{\sqrt{n}}\sum_{i=1}^n(F_{Y_j|X}(y|X^*_i)-F_{Y_j^*}(y))\right|\\
&+c_0\sup_{y\in\Y^*_j}
\left|\frac{1}{\sqrt{n}}\sum_{i=1}^n\frac{f_{X^*}(X_i)}{f_{X}(X_i)}
(\1\{Y_{ji}\leq y\}-F_{Y_j|X}(y|X_i))\right|+o_p(1).
\end{align*}
Note that for each $y\in\Y^*_j$,
\[
\frac{1}{n}\sum_{i=1}^n(F_{Y_j|X}(y|X^*_i)-F_{Y_j^*}(y))\quad\text{and}\quad\frac{1}{n}\sum_{i=1}^n\frac{f_{X^*}(X_i)}{f_{X}(X_i)}(\1\{Y_{ji}\leq y\}-F_{Y_j|X}(y|X_i))
\]
are both degenerate U statistics of order 1.

We can show that under Assumption~\ref{S6}, the class $\{F_{Y^*_j}(y):y\in\Y^*_j\}$ is Euclidean by Lemma 2.13 of \cite{PakesPollard1989}.
In addition, the arguments in \cref{LEM2} \textit{mutatis mutandis} show that both classes
$\{\tilde{x}\mapsto F_{Y_j|X}(y|\tilde{x})-F_{Y^*_j}(y):y\in\Y^*_j\}$
and $\{(\tilde{x},\tilde{y})\mapsto \frac{f_{X^*}(\tilde{x})}{f_{X}(\tilde{x})}(\1\{\tilde{y}\leq y\}-F_{Y_j|X}(y|\tilde{x})): y\in\Y^*_j\}$ are Euclidean for some envelope.
Hence, the desired result is achieved because
\[
\sup_{y\in\Y^*_j}\left|\frac{1}{\sqrt{n}}\sum_{i=1}^n
(F_{Y_j|X}(y|X^*_i)-F_{Y_j^*}(y))\right|=o_p(1)
\]
and
\[
\sup_{y\in\Y^*_j}
\left|\frac{1}{\sqrt{n}}\sum_{i=1}^n\frac{f_{X^*}(X_i)}{f_{X}(X_i)}
(\1\{Y_{ji}\leq y\}-F_{Y_j|X}(y|X_i))\right|=o_p(1)
\]
by Corollary 4 of \cite{Sherman1994}.
\end{proof}

\begin{lem}
\label{LEM3}
Suppose that the assumptions of \cref{PRO3} hold.
\begin{enumerate}[(i)]
\item Let $\mathcal{L}$ be the function defined in \cref{Lfun}. The class of functions
\begin{align*}
\{(\tilde{w}_1,\tilde{w}_2)\mapsto\mathcal{L}(\tilde{w}_1,\tilde{w}_2;y_1,y_2,h): (y_1,y_2)\in \Y_1\Y_2, h>0\}
\end{align*}
is Euclidean for some envelope.
\item Let $\mathcal{M}$ be the function defined in \cref{Mfun}. The class of functions
\begin{align*}
\{(\tilde{w}_1,\tilde{w}_2)\mapsto\mathcal{M}(\tilde{w}_1,\tilde{w}_2;y_1,y_2,h): (y_1,y_2)\in \Y_1\Y_2, h\in(0,1)\}
\end{align*}
is Euclidean for some envelope.
\end{enumerate}
\end{lem}

\begin{proof}
\leavevmode
\begin{enumerate}[(i)]
\item
By Lemma 2.14 of \cite{PakesPollard1989}, it is sufficient to show the classes $\{(\tilde{y}_1,\tilde{y}_2)\mapsto\1\{\tilde{y}_1\leq y_1,\tilde{y}_2\leq y_2\}: (y_1,y_2)\in \Y_1\Y_2\}$,
$\{\tilde{x}\mapsto F_{Y_1Y_2|X}(y_1,y_2|\tilde{x}): (y_1,y_2)\in \Y_1\Y_2\}$, $\{(\tilde{x}_1,\tilde{x}_2)\mapsto K\left((\tilde{x}_1-\tilde{x}_2)/h\right):h>0\}$, and $\{\tilde{x}\mapsto h^d\bar{f}(\tilde{x};h):h>0\}$ are all Euclidean. First, Example 2.6.1 and Theorem 2.6.7 of \cite{vanderVaartWellner1996} imply that the class $\{(\tilde{y}_1,\tilde{y}_2)\mapsto\1\{\tilde{y}_1\leq y_1,\tilde{y}_2\leq y_2\}: (y_1,y_2)\in \Y_1\Y_2\}$ is Euclidean for a constant envelope.
Next, from Lemma 2.13 of \cite{PakesPollard1989} and under Assumption~\ref{S4},
the class $\{\tilde{x}\mapsto F_{Y_1Y_2|X}(y_1,y_2|\tilde{x}): (y_1,y_2)\in \Y_1\Y_2\}$ is Euclidean for some envelope.
In addition, the discussion on page 911 of \cite{Gine2002} implies that $\{\tilde{u}\mapsto K(\tilde{u}/h):h>0\}$ is a VC-subgraph class under Assumption~\ref{K1}.
It follows from Lemma 2.6.18(vii) of \cite{vanderVaartWellner1996} that the class $\{(\tilde{x}_1,\tilde{x}_2)\mapsto K\left((\tilde{x}_1-\tilde{x}_2)/h\right):h>0\}$ is a VC-subgraph class and thus Euclidean for a constant envelope.
Finally, Lemma 5 of \cite{Sherman1994} implies that the class $\{\tilde{x}\mapsto h^d\bar{f}(\tilde{x};h):h>0\}$ is also Euclidean for a constant envelope.
\item
By the argument in part (i) and Lemma 5 of \cite{Sherman1994},
it suffices to show that $\{(\tilde{x}_1,\tilde{x}_2)\mapsto h^dQ(\tilde{x}_1,\tilde{x}_2;h):h\in(0,1)\}$ is Euclidean. Let
\begin{align*}
V(x_1,x_2,z;h)\equiv\frac{f_{X^*}(x_1+hz)}{[f_{X}(x_1+hz)]^2}K\left(\frac{x_1-x_2}{h}+z\right)
-\frac{f_{X^*}(x_1)}{[f_{X}(x_1)]^2}K\left(\frac{x_1-x_2}{h}\right).
\end{align*}
Note that for any $\tilde{h}_1$ and $\tilde{h}_2$ in $(0,1)$ and random vectors $\tilde{X}_1$ and $\tilde{X}_2$,
\begin{align*}
&\E[|\tilde{h}_1^dQ(\tilde{X}_1,\tilde{X}_2;\tilde{h}_1)
-\tilde{h}_2^dQ(\tilde{X}_1,\tilde{X}_2;\tilde{h}_2)|]\\
\leq{}& 2^d\sqrt{\E[|K(U)|^2]}\sqrt{\E[|V(\tilde{X}_1,\tilde{X}_2,U; \tilde{h}_1)-V(\tilde{X}_1,\tilde{X}_2,U;\tilde{h}_2)|^2]}
\end{align*}
where $U$ is a random vector uniformly distributed on $[-1,1]^d$.
Hence, it remains to show that the two classes
\begin{align*}
\left\{(\tilde{x},\tilde{z})\mapsto \frac{f_{X^*}(\tilde{x}+h\tilde{z})}{[f_{X}(\tilde{x}+h\tilde{z})]^2}:h\in(0,1)\right\}
\end{align*}
and
\begin{align*}
\left\{(\tilde{x}_1,\tilde{x}_2,\tilde{z})\mapsto K\left(\frac{\tilde{x}_1-\tilde{x}_2}{h}+\tilde{z}\right):h\in(0,1)\right\}
\end{align*}
are both Euclidean. It follows from Lemma 2.13 of \cite{PakesPollard1989} that the former class is Euclidean because
for any $\tilde{h}_1$ and $\tilde{h}_2$ in $(0,1)$,
\begin{align*}
\sup_{\tilde{x},\tilde{z}}
\left|\frac{f_{X^*}(\tilde{x}+\tilde{h}_1\tilde{z})}{[f_{X}(\tilde{x}+\tilde{h}_1\tilde{z})]^2}
-\frac{f_{X^*}(\tilde{x}+\tilde{h}_2\tilde{z})}
{[f_{X}(\tilde{x}+\tilde{h}_2\tilde{z})]^2}\right|\leq c|\tilde{h}_1-\tilde{h}_2|
\end{align*}
for some constant $c$ by Assumptions~\ref{S1} and~\ref{S3}.
Finally, Lemma 22(i) of \cite{Nolan1987} implies that the latter class is Euclidean by Assumptions~\ref{K1} and~\ref{K4}.\qedhere
\end{enumerate}
\end{proof}

\end{appendices}

\normalsize
\bibliography{CFcopula_20230309}
\bibliographystyle{ecta}

\end{document}

%% file: preamble.tex
\usepackage{mathtools,amssymb,amsthm,bm,bbm}
    \allowdisplaybreaks[1]
\usepackage{microtype,xcolor,natbib}
\usepackage{amssymb,amsthm,bm,bbm,microtype,xcolor,natbib}
\usepackage[shortlabels]{enumitem}
\usepackage[margin=1in,a4paper]{geometry}
\usepackage[onehalfspacing]{setspace}
\usepackage[colorlinks,allcolors=blue!50!black]{hyperref}
\usepackage[capitalise,noabbrev]{cleveref}
\usepackage[title]{appendix}

\usepackage{booktabs,multirow,graphicx,epstopdf,dcolumn}
\newcolumntype{.}[1]{D{.}{.}{#1}}
\newcolumntype{d}[1]{D{.}{.}{#1}}
\usepackage[flushleft]{threeparttable}
\usepackage[labelformat=simple]{subcaption}

\usepackage[labelsep=period]{caption}

\makeatletter
\newcommand{\pushleft}[1]{\ifmeasuring@#1\else\omit$\displaystyle#1$\hfill\fi\ignorespaces}
\newcommand{\pushright}[1]{\ifmeasuring@#1\else\omit\hfill$\displaystyle#1$\fi\ignorespaces}
\makeatother

\let\originalleft\left
\let\originalright\right
\renewcommand{\left}{\mathopen{}\mathclose\bgroup\originalleft}
\renewcommand{\right}{\aftergroup\egroup\originalright}

\DeclareFontFamily{U}{mathx}{\hyphenchar\font45}
\DeclareFontShape{U}{mathx}{m}{n}{<-> mathx10}{}
\DeclareSymbolFont{mathx}{U}{mathx}{m}{n}
\DeclareMathAccent{\widebar}{0}{mathx}{"73}

\DeclareMathOperator{\1}{\operatorname{1}}
\DeclareMathOperator{\E}{\operatorname{E}}

\newcommand{\X}{\mathcal{X}}
\newcommand{\Y}{\mathcal{Y}}
\newcommand{\R}{\mathbb{R}}
\renewcommand{\d}[1]{\ensuremath{\operatorname{d}\!{#1}}}

\theoremstyle{plain}

\newtheorem{lem}{Lemma}
\newtheorem{thm}{Theorem}

\newtheorem{pro}{Proposition}
\theoremstyle{definition}

\newlist{asmenum}{enumerate}{1}
\setlist[asmenum]{label=(\roman*),ref=\theasm{}(\roman*)}
\crefalias{asmenumi}{asm}
\crefname{asm}{Assumption}{Assumptions}
\crefname{lem}{Lemma}{Lemmas}
\crefname{thm}{Theorem}{Theorems}
\crefname{cor}{Corollary}{Corollaries}
\crefname{pro}{Proposition}{Propositions}
\crefname{rmk}{Remark}{Remarks}
\crefname{equation}{}{}
 